\documentclass[12pt, twoside]{article}
\usepackage{amsmath,amsthm,amssymb}
\usepackage{times}
\usepackage{enumerate}

\theoremstyle{definition}
\newtheorem{thm}{Theorem}[section]

\newtheorem{lem}[thm]{Lemma}

\newtheorem{rem}[thm]{Remark}

\numberwithin{equation}{section}

\newcommand{\subjclass}[1]{\bigskip\noindent\emph{2010 Mathematics Subject Classification:}\enspace#1}
\newcommand{\keywords}[1]{\noindent\emph{Keywords:}\enspace#1}
\newcommand{\bA}{\mathbf A}
\newcommand{\bB}{\mathbf B}

\newcommand{\bH}{\mathbf H}
\newcommand{\bI}{\mathbf I}

\newcommand{\bP}{\mathbf P}

\newcommand{\ba}{\mathbf a}
\newcommand{\bb}{\mathbf b}

\newcommand{\bn}{\mathbf n}
\newcommand{\be}{\mathbf e}

\newcommand{\bT}{\mathbf T}
\newcommand{\bu}{\mathbf u}
\newcommand{\bv}{\mathbf v}

\newcommand{\bz}{\mathbf z}

\newcommand{\Div}{\mathop{\rm div}}
\newcommand{\divG}{{\mathop{\,\rm div}}_{\Gamma}}
\newcommand{\gradG}{\nabla_{\Gamma}}

\newcommand{\cL}{\mathcal L}

\newcommand{\Gs}{\mathcal{S}} %{\Gamma_\ast}

\renewcommand{\div}{\textrm{div}\ \!}

\newcommand{\Rn}{\Bbb{R}^n}

\newcommand{\tr}{{\rm tr}}

\newcommand{\bsigma}{\boldsymbol{\sigma}}
\newcommand{\bphi}{\boldsymbol{\phi}}

\newcommand{\bxi}{\mbox{\boldmath$\xi$\unboldmath}}

\begin{document}

%%%%% To ease editing, add:

%\baselineskip=17pt

%%%%%%%%%%%%%%%%

\title{Incompressible fluid problems on embedded surfaces: Modeling and variational formulations\thanks{Interfaces and Free Boundaries 20 (2018), 353--378~~DOI 10.4171/IFB/405}}

\author{Thomas Jankuhn \\
Institut f\"ur Geometrie und Praktische  Mathematik \\
RWTH-Aachen University, D-52056 Aachen, Germany\\
jankuhn@igpm.rwth-aachen.de\\[1ex]
Maxim A. Olshanskii\\
Department of Mathematics, University of Houston, Houston, Texas 77204\\
molshan@math.uh.edu\\[1ex]
Arnold Reusken\\
Institut f\"ur Geometrie und Praktische  Mathematik \\
RWTH-Aachen
University, D-52056 Aachen, Germany\\
reusken@igpm.rwth-aachen.de}

\date{}

\maketitle

%%%%%%%%

\begin{abstract}
Governing equations of motion for a viscous incompressible  material  surface are derived from the balance laws of continuum mechanics.
The surface is treated as a time-dependent  smooth  orientable manifold of codimension one in an ambient Euclidian space. We use elementary tangential calculus to derive the governing equations in terms of exterior differential operators in  Cartesian coordinates. The resulting equations can be seen as the Navier-Stokes equations posed on an evolving manifold.
We consider a splitting of the surface Navier-Stokes system into   coupled equations for the tangential and normal motions of the material surface.
We then restrict ourselves to the case of a geometrically stationary manifold of codimension one embedded in $\Bbb{R}^n$.  For this case, we present new well-posedness results for the simplified surface fluid model consisting of the surface Stokes equations. Finally, we propose and analyze several alternative variational formulations for this  surface Stokes problem, including constrained and penalized formulations, which are convenient for Galerkin discretization methods.


\end{abstract}

\section{Introduction}
Fluid equations on manifolds appear in the literature on  mathematical modelling  of emulsions, foams and  biological  membranes, e.g.~\cite{slattery2007interfacial,brenner2013interfacial,nitschke2012finite,reuther2015interplay}; they are also studied as a mathematical problem of its own interest, e.g.~\cite{ebin1970groups,Temam88,taylor1992analysis,arnol2013mathematical,mitrea2001navier,arnaudon2012lagrangian}.
In certain applications, such as the dynamics of liquid membranes~\cite{arroyo2009}, one is interested in formulations of fluid equations on evolving (time-dependent) surfaces. Such equations are considered in several places in the literature. The authors of \cite{arroyo2009} formulate a continuum model of
fluid membranes embedded in a bulk fluid, which includes  governing equations for a two-dimensional viscous fluid moving on a curved, time-evolving surface.  The derivation of a surface strain tensor in that paper uses techniques and notions  from differential geometry ($k$-forms). A similar model was derived from  balance laws for mass and momentum and associated constitutive equations in \cite{rangamani2013interaction}. The derivation and the resulting model uses intrinsic variables on a surface. Equations for surface fluids in the context of
two-phase flow are derived or used in \cite{BothePruess2010,barrett2014stable,nitschke2012finite,ReuskenZhang}. In those papers the surface fluid dynamics is strongly coupled through a no-slip condition with the bulk fluid dynamics. An energetic variational approach was recently used in \cite{Gigaetal} to derive the dynamical
system for the motion of an incompressible viscous fluid on  an evolving surface.

Computational methods and numerical analysis of these methods for fluid equations on surfaces is a relatively new field of research. Exploring the line of research starting from the seminal paper \cite{scriven1960dynamics}, it is noted in \cite{arroyo2009} and \cite{nitschke2012finite} that ``the equations of motion are formulated intrinsically in a two-dimensional manifold with time-varying metric and make extensive use of the covariant derivative and calculations in local coordinates, which involve the coefficients of the Riemannian connection and its derivatives. The complexity of the equations may explain why they are often written but never solved for arbitrary surfaces.'' Recent research addressing the numerical solution of fluid equations on surfaces includes   \cite{nitschke2012finite,rangamani2013interaction,rahimi2013curved,barrett2014stable,reuther2015interplay,rodrigues2015semi,ReuskenZhang}.

We discuss the  two main contributions of this paper. The first one is related to modeling.  Based on fundamental surface continuum mechanical principles treated in  \cite{GurtinMurdoch75,MurdochCohen79} we derive fluid equations on an evolving surface from the  conservation laws of mass and momentum for a viscous  material  surface embedded in an ambient continuum medium. We assume that the bulk medium interacts with the fluidic membrane through the area forces.  To derive the governing equations, we use \emph{only elementary tangential differential calculus} on a manifold. As a result, the surface PDEs that we derive are formulated in terms of   differential operators in the Cartesian coordinates. In particular, we %do not use notions from differential geometry and
avoid  the use of local coordinates. Using tangential differential operators makes the formulation more convenient for numerical purposes and facilitates the application of a level set method or other implicit surface representation techniques (no local coordinates or parametrization involved) to describe the surface evolution.
The resulting equations can be seen as the Navier-Stokes equations for a viscous incompressible 2D surface fluid posed on an evolving  manifold  embedded in  $\Bbb{R}^3$. The same equations have been derived and studied in the recent paper \cite{Gigaetal}. In that paper, however, the derivation is based on global energy principles instead of local conservation laws.
For gaining some further insight in this rather complex surface Navier-Stokes model, we  consider a splitting of the  system into   coupled equations for the tangential and normal motions of the material surface.
%It occurs that the dynamics of tangential fluid velocities is driven by a (quasi-)parabolic equation, while the equation for the evolution of the normal velocity involves only first order derivatives. The splitting can be a natural departing point for numerical integration of the surface Navier-Stokes system.
 The resulting equation for tangential motions agrees with one derived in \cite{Gigaetal}, but  differs from the one found in~\cite{arroyo2009}.
We comment on how the surface Navier-Stokes equations that we consider are related to other formulations of surface fluid equations found in the literature (Remarks~\ref{remArroyo} and Section~\ref{s_other}).

The second main contribution of this paper is a derivation of well-posedness results for a strongly simplified case.
We  restrict ourselves to  a geometrically stationary closed smooth manifold of codimension one, embedded in $\Bbb{R}^n$.   For this case, we present new well-posedness results for the surface Stokes equations.  Key ingredients in the analysis are a surface Korn's inequality and an inf-sup result for the Stokes bilinear form that couples surface pressure and surface velocity. We propose and analyze several different variational formulations of the  surface Stokes problem, including constrained and penalized formulations, which are convenient for Galerkin discretization methods.

The remainder of this paper is organized as follows. Section~\ref{s_prelim} collects necessary preliminaries and auxiliary results. In section~\ref{sectmodel} we derive the governing equations for the motion of
a viscous  material  surface, the surface Navier-Stokes system.  We also consider a directional splitting of the system and discuss alternative formulations of the surface fluid equations.
In section~\ref{s_Stokes} we prove a fundamental surface Korn's inequality and   well-posedness
of a variational formulation of the surface Stokes problem.  In sections~\ref{SecLagrange} and~\ref{s_aug} we introduce alternative weak formulations of the surface Stokes problem, which we believe are more convenient for Galerkin discretization methods such as surface finite element methods.

\section{Preliminaries}\label{s_prelim} This section recalls some basics of tangential calculus for evolving manifolds of codimension one. Several helpful auxiliary results are also proved in this section.
Consider $\Gamma(t) \subset \Bbb{R}^n$, $n \geq 3$, a $(n-1)$-dimensional closed, smooth, simply connected evolving manifold for $t\ge0$. We are mainly interested in $n=3$, but most of the analysis applies for general $n$. In the modeling part, section~\ref{sectmodel}, we only consider $n=3$. {  Concerning the smoothness conditions for $\Gamma(t)$ we note that it will sufficient to assume that  for any given $t\ge0$ the surface $\Gamma(t)$ has $C^3$ smoothness. In the remainder we always assume that this holds. For $k \in {0,1,2,}$ the spaces $C^k(\Gamma)$ are defined in the usual way via charts.}  The fact that the manifold is embedded in $\Bbb{R}^n$ plays a key role in the derivation and formulation of the PDEs. For example, for a $C^3$ manifold, a normal extension of $f\in C^k(\Gamma)$,  $k \in {0,1,2,}$ is a $C^k$-smooth function in a $\Bbb{R}^n$-neighborhood of $\Gamma$ and the surface differential operators  can be formulated in terms of differential operators in Euclidean space $\Bbb{R}^n$, with respect to the standard basis in $\Bbb{R}^n$.

The outward pointing normal vector on $\Gamma=\Gamma(t)$ is denoted by $\bn=\bn(x,t)$, and $\bP=\bP(x,t)=\bI- \bn\bn^T$ is the normal projector on the tangential space at $x\in\Gamma(t)$.
First we consider $\Gamma=\Gamma(t)$ for some fixed $t$ and introduce spatial differential operators. For $f:\, \Bbb{R}^n \to \Bbb{R}^{m}$, we denote by $\nabla f(x) \in L(\Bbb{R}^n,\Bbb{R}^{m})$  the Frechet derivative at $x \in \Bbb{R}^n$, where $L(\Bbb{R}^n,\Bbb{R}^{m})$ is the vector space of linear transformations from $\Bbb{R}^n$ to $\Bbb{R}^{m}$. We often skip the argument $x$ in the notation below.  The partial derivative is denoted by $\partial_i f= (\nabla f) \be_i \in \Bbb{R}^m$, $i=1,\ldots,n$.  Hence $(\nabla f) \bz= \sum_{j=1}^n \partial_j f z_j$ for $\bz \in \Bbb{R}^n$.
Note that for a scalar function $f$, i.e.,  $m=1$, $\nabla f$ is a \textit{row} vector. {  In the setting of this paper it is convenient to use this less standard row (instead of column) representatian for the gradient (e.g., the  formula $\nabla_\Gamma f =(\nabla f)\bP$  holds for the tangential gradient, cf.~\eqref{defdeli1}).  The vector $\nabla^T f:=(\nabla f)^T$ denotes the column gradient vector}. %Hence in Stokes etc. we have to write $\nabla p^T$; it has to be checked, whether this is consistently done}

The \emph{tangential} derivative (along $\Gamma$) is defined as $(\nabla f) \bP \bz= \sum_{j=1}^n \partial_j f(\bP \bz)_j$ for $\bz \in \Bbb{R}^n$. For $m=1$, i.e, $f:\Bbb{R}^n \to \Bbb{R}$ the corresponding $i$-th (tangential) partial derivative is denoted by $\nabla_i$:
\begin{equation} \label{defdeli1}
 \nabla_i f= \sum_{j=1}^n \partial_j f(\bP \be_i)_j, ~~\text{and}~~\nabla_\Gamma f:=\big(\nabla_1f, \ldots, \nabla_n f\big) = (\nabla f)\bP.
\end{equation}
We also need such {covariant partial derivatives} for $m=n$ and $m=n\times n$. For $m=n$ the $i$-th covariant partial derivative of $\bv: \Rn \to \Rn$ is defined as
\begin{equation} \label{defdeli2}
 \nabla_i \bv= \sum_{j=1}^n \bP \partial_j  \bv (\bP \be_i)_j, ~~\text{and}~~\nabla_\Gamma \bv:=\big(\nabla_1 \bv\ldots \nabla_n \bv\big) = \bP (\nabla\bv) \bP.
\end{equation}
We shall use the notation $\nabla_\Gamma^T f:=(\nabla_\Gamma f)^T$, $\nabla_\Gamma^T \bv:=(\nabla_\Gamma \bv)^T$ for the transposed vector and matrix, and similarly for $\nabla_\Gamma$ replaced by $\nabla$.  For $m=n\times n$ the $i$-th covariant partial derivative of $\bA: \Rn \to \Bbb{R}^{n\times n}$ is defined as
\begin{equation} \label{defdeli3}
 \nabla_i \bA= \sum_{j=1}^n \bP \partial_j \bA \bP (\bP \be_i)_j, ~~\text{and}~~\nabla_\Gamma \bA:=\big(\nabla_1 \bA \ldots \nabla_n \bA\big).
\end{equation}
Note that from $\bn^T\bP =\bP \bn=0$ we get $\bP (\partial_j \bP)\bP= -\bP(\partial_j \bn \bn^T + \bn \partial_j \bn^T)\bP= 0$, hence $\nabla_i\bP=0$, $i=1,\ldots, n$, i.e., $\gradG \bP=0$. The covariant partial derivatives of $f$, $\bv$, or $\bA$ depend only
on the values of these fields on $\Gamma$. For scalar functions $f,g$ and vector functions $\bu,\bv: \Gamma\to \Bbb{R}^n$ we have the following product rules:
\begin{align}
 \nabla_\Gamma(fg) & = g  \nabla_\Gamma f +  f \nabla_\Gamma g \\
 \nabla_\Gamma(\bu\cdot\bv) & = \bv^T\nabla_\Gamma  \bu + \bu^T\gradG \bv, \quad { \text{if}~~\bP\bu=\bu,~\bP\bv=\bv,} \\
 \gradG(f \bu) &= f \gradG \bu + \bP \bu \gradG f. \label{id3}
\end{align}
%(note the $\bP$ used in \eqref{id3}).
Besides these covariant  derivatives we also need \emph{tangential divergence operators} for $\bv: \Gamma \to \Rn$ and $\bA: \Gamma \to \Bbb{R}^{n\times n}$.  These are defined as follows:
\begin{align}
 \divG \bv & := \tr (\gradG \bv)= \tr (\bP (\nabla\bv) \bP)=\tr (\bP (\nabla\bv)))=\tr ((\nabla\bv) \bP), \label{defdiv1} \\
% \divG \bA  & := \begin{pmatrix} \divG (e_1^T \bA) \\
%               \vdots \\
%               \divG (e_n^T \bA)
%              \end{pmatrix}
 \divG \bA  & := \left( \divG (\be_1^T \bA),\,
               \dots, \,
               \divG (\be_n^T \bA)\right)^T.
                \label{defdiv2}
\end{align}
These tangential differential operators will be used in the modeling of conservation laws in section~\ref{sectmodel}. %resulting in surface partial differential equations in which these operators occur.
In particular, %in the   momentum equation on $\Gamma$,
the  differential operator  $\bP \divG\big(\gradG \bv+\gradG^T \bv\big) $, which is the tangential analogon of the $\Div (\nabla\bv +\nabla^T\bv)$ operator in Euclidean space, plays a key role. We derive some properties of this differential operator.

We first relate $\bP \divG(\gradG \bv)$ to a Laplacian. For this  we introduce the space of smooth  tangential vector fields $C^k_T(\Gamma)^n:=\{\, \bv \in C^k(\Gamma)^n~|~\bP\bv =\bv \,\}$, with scalar product $(\bu,\bv)_0=\int_\Gamma \bu\cdot\bv \, ds$, and the space of  smooth tangential tensor fields  $C^k_T(\Gamma)^{n\times n}:=\{\, \bA \in C^k(\Gamma)^{n \times n}~|~\bP\bA \bP =\bA\,\}$, with scalar product $(\bA,\bB)_0:=\int_{\Gamma} \tr (\bA \bB^T) \, ds$. From the partial integration identity (see, e.g.,  (14.17) in \cite{GReusken2011}),
\[ \begin{split}
  \int_{\Gamma} \bv\cdot (\bP\divG \bA) \, ds & = \int_{\Gamma} \bv\cdot \divG \bA \, ds \\ &  = - \int_\Gamma \tr(\bA^T \gradG \bv)\,ds, \quad\bv \in C^1_T(\Gamma)^n, ~ \bA \in C^1_T(\Gamma)^{n\times n},
\end{split}
\]
it follows that for $\cL: C^1_T(\Gamma)^{n\times n}  \to C^0_T(\Gamma)^n$ given by $\cL(\bA)= \bP \divG (\bA)$ we have
\[
  (\cL(\bA),\bv)_0= -(\bA,\gradG \bv)_0 \quad \text{for all}~~\bv \in C^1_T(\Gamma)^n,~\bA \in C^1_T(\Gamma)^{n\times n}.
\]
Hence, $-\cL$ is the adjoint of $\gradG$, i.e., $\cL= - \gradG^\ast$. Thus we have
\begin{equation} \label{Laplacian}
  \bP \divG(\gradG \bv)= \cL(\gradG \bv)= - \gradG^\ast \gradG \bv=:  \Delta_\Gamma \bv.
\end{equation}
This \emph{vector Laplacian}  $\Delta_\Gamma$ is the so-called Bochner Laplacian \cite{rosenberg1997laplacian}.
It can be extended to a self-adjoint operator on a suitable space of vector fields on $\Gamma$.
%It follows from the Poincare inequality $\|\bv\|_0 \leq c \|\gradG \bv\|_0$, $\bv\in  C^\infty_T(\Gamma)^n$, see, e.g.,  Lemma 5.3 in \cite{hansbo2016analysis}, that  $\Delta_\Gamma$ is {negative} definite.

The mapping $\bv \to \bP \divG\gradG^T \bv$ requires more calculations. Note that in Euclidean space we have $ \div(\nabla^T \bv)_i= \div(\be_i^T \nabla^T\bv)= \div(\partial_i\bv)= \partial_i( \div \bv)$. Hence, for divergence free functions $\bv$ we have $ \div\nabla^T   \bv=0$. For the corresponding surface differential operator we do not have a simple commutation relation, and the analysis becomes more complicated.  In \cite{arroyo2009,nitschke2012finite} this mapping is analyzed with intrinsic tools of differential geometry. It is, however, not clear how the divergence operators used in those papers, which are defined via differential forms, are related to the tangential divergence operator $\divG$ introduced above, which is defined in Euclidean space $\Rn$.  Lemma~\ref{LemGauss} below shows  a useful representation for $\bP \divG\gradG^T \bv$.  The proof of the lemma is given in the Appendix and it only uses elementary  tangential calculus.
For a vector field $\bv$  on $\Gamma(t)$ we shall use throughout the paper the notion $\bv_T=\bP\bv$  for the tangential part and
$v_N=\bv\cdot\bn$ for the normal coordinate, so that
\[
\bv=\bv_T+v_N\bn\quad\text{on}~\Gamma(t).
\]
\begin{lem} \label{LemGauss}
Let $\bH=\nabla_\Gamma\bn \in \Rn$ be the  Weingarten mapping (second fundamental form) on $\Gamma(t)$ and $\kappa:=\tr(\bH)$ the (doubled) mean curvature. The following holds:
\begin{align} \label{Res1}
 \bP \divG\gradG^T \bv & = \gradG^T \divG \bv+ \big(\tr(\bH)\bH-\bH^2\big) \bv,\quad \forall~\bv \in C^2_T(\Gamma)^n, \\
  \bn\cdot\divG\gradG^T \bv &= \bn\cdot \divG(\gradG \bv) \nonumber = - \tr(\bH \gradG \bv)\\ & =- \tr(\bH \gradG \bv_T)-v_N \tr(\bH^2),\quad \forall~\bv \in C^2(\Gamma)^n, \label{Res1a} \\
 \bP \divG (\bH) &= \gradG^T \kappa. \label{newres}
\end{align}
If $n=3$, then \eqref{Res1} simplifies to
\begin{equation} \label{Res2}
 \bP \divG\gradG^T \bv= \gradG^T \divG \bv+ K\bv,\quad \forall~\bv \in C^2_T(\Gamma)^3,
\end{equation}
where $K$ is the Gauss curvature, i.e. the product of the two principal curvatures.
%, which are non-zero  eigenvalues of $\bH$.
\end{lem}

We now introduce some notations related to the evolution of $\Gamma(t)$ in time. Let $\Gs$ be the  $n$-dimensional manifold defined by the evolution of $\Gamma$,
 \[\Gs:=\bigcup\limits_{t>0}\Gamma(t)\times\{t\};\] the (space--time) manifold   $\Gs$ is embedded in $\mathbb{R}^{n+1}$.
  For the rest of the paper, we assume that $\Gs$ is $C^2$ smooth (and we continue to assume that $\Gamma(t)$ is $C^3$ smooth for any fixed $t\ge0$).
We  assume a flow field $\bu:\, \Rn\to \Rn$ such that $V_\Gamma=\bu\cdot \bn$ on $\Gs$,  where $V_\Gamma$ denotes the \textit{normal velocity of $\Gamma$}. For a smooth $f:\,\Rn\to\mathbb{R}$ we  consider the material derivative $\dot f$ (the derivative along material trajectories  in the velocity field $\bu$).
\[
\dot f = \frac{\partial f}{\partial t} + \sum_{i=1}^n \frac{\partial f}{\partial x_i} u_i=\frac{\partial f}{\partial t}+ (\nabla f)  \bu .
\]
The material derivative $\dot f$ is a tangential derivative for $\Gs$, and hence it depends only on the surface values of $f$ on $\Gamma(t)$. For a vector field $\bv$, we define $\dot\bv$ componentwise, i.e., $\dot\bv= \frac{\partial \bv}{\partial t} + (\nabla \bv)\bu$.
In Lemma~\ref{ident1} we derive some useful identities for the material derivative of the normal vector field and normal projector  on $\Gamma$. %that  we need in section~\ref{splitting}.

\begin{lem} \label{ident1}
 The following identities hold on $\Gamma(t)${\rm:}
 \begin{align} \label{id1}
  \dot \bn & = \bH \bu_T - \gradG^T u_N  =- \bP(\nabla^T \bu)\, \bn , \\
  \dot \bP & = \bP(\nabla^T \bu) (\bI-\bP) +(\bI-\bP)(\nabla \bu) \bP. \label{id1a}
 \end{align}
\end{lem}
\begin{proof}
Let $d(x,t)$ be the signed distance function to $\Gamma(t)$ defined in a neighborhood $U_t$ of $\Gamma(t)$.
Define the normal extension of $\bn$ and $\bH$ to $U_t$ by $\bn^T=\nabla d$, $\bH=\nabla^2 d$ (note that the latter identity shows that $\bH$ is symmetric), and consider the  closest point projection $p(x,t)=x -d(x,t) \bn(x,t)$, $x \in U_t$. We then have
\[
  \frac{\partial d}{\partial t}(x,t) = - u_N\big(p(x,t),t\big), \quad x \in U_t.
\]
Using  the chain rule we get
\[
  \nabla [u_N(p(x,t),t)] = \nabla_\Gamma u_N (p(x,t),t)\big(\bI - d(x,t)\bH\big) \quad x \in U_t.
\]
Take $x \in \Gamma(t)$, using $d(x,t)=0$, $p(x,t)=x$ and $\nabla d =\bn^T$, we obtain
\begin{equation} \label{id2}
 \frac{\partial \bn^T}{\partial t}= \frac{\partial}{\partial t} \nabla d= \nabla\frac{\partial d}{\partial t} = - \nabla_\Gamma u_N, \quad \text{on}~~  \Gamma(t).
\end{equation}
Using this and $\bH \bn=0$ we get
\[
 \dot \bn = \frac{\partial \bn}{\partial t}+ (\nabla \bn) \bu= - \gradG^T u_N +  \bH \bu_T,
\]
which is the first identity in \eqref{id1}.
  From  $\bu_T \cdot \bn=0$  we get $\bn^T \nabla \bu_T=- \bu_T^T \nabla \bn $ and combined with the symmetry of $\bH$ we get
  \begin{equation}\label{aux1}
  \bH \bu_T=- (\nabla^T \bu_T)\bn.
  \end{equation}
  Furthermore, we note that $\nabla(u_N \bn)= \bn \nabla u_N  +u_N \nabla \bn$, hence $\bn^T \nabla(u_N \bn)=\nabla u_N$. Using this, the result in \eqref{aux1} and $\bP\bH=\bH$ we get
\begin{align*}
  - \gradG^T u_N +  \bH \bu_T & = -\bP\big( \nabla^T u_N + (\nabla^T \bu_T) \bn \big) \\ & =-\bP\big(\nabla^T(u_N \bn) \bn  + (\nabla^T \bu_T) \bn \big) = - \bP (\nabla^T \bu) \bn,
\end{align*}
which is the second identity in \eqref{id1}. The result in \eqref{id1a} immediately follows from $\dot \bP= -\dot \bn \bn^T - \bn \dot \bn^T$ and the second identity in \eqref{id1}.
\end{proof}
\smallskip

From \eqref{id1} we see that the vector field $\dot\bn$ is always tangential to $\Gamma(t)$.

\section{Modeling of  material surface flows} \label{sectmodel}
In this section, we assume $\Gamma(t)$ is a \textit{material} surface (fluidic membrane) embedded in $\Bbb{R}^3$ as defined in \cite{GurtinMurdoch75,MurdochCohen79}, with density distribution $\rho(x,t)$.
By $\bu(x,t)$, $x \in \Gamma(t)$,  we denote the smooth velocity field  of the density flow on $\Gamma$, i.e. $\bu(x,t)$ is the velocity of the material point  $x\in\Gamma(t)$. The geometrical evolution of the surface is defined by the normal velocity $u_N$, for $\bu(x,t)=\bu_T+u_N\bn$.  For all $t\in[0,T]$, we assume $\Gamma(t) \subset \Bbb{R}^3$ to be smooth, closed and embedded in an ambient continuum medium, which exerts external (area) forces on the material surface.

%{\bf I am not satisfied with the modeling part, yet. cf. Remarks below and at the end of this section}

%The covariant derivatives and tangential gradient used in \cite{GurtinMurdoch75,MurdochCohen79} are the same as the ones  introduced above in \eqref{defdeli1}-\eqref{defdiv2}.
%Below we derive the surface Navier--Stokes equations, which describe the dynamics \emph{both} of the normal velocity of the manifold (i.e., $u_N$) and of the density flow over $\Gamma(t)$ (i.e., $\bu_T$).

Let $\gamma(t) \subset \Gamma(t)$ be a material subdomain.
 For a smooth $f:\Gs\to \mathbb{R}$, we shall make use of  the  transport formula (also known as a Leibniz rule; see, e.g., \cite{DEreview} Theorem 5.1),
\begin{equation} \label{Leibniz}
 \frac{d}{dt} \int_{\gamma(t)} f \, ds= \int_{\gamma(t)}( \dot f + f\divG \bu )\, ds.
\end{equation}
For a smooth tangential field $\bu_T:\Gs\to \mathbb{R}^n$ we need the surface Stokes formula (see, e.g. \cite{GReusken2011}, section~14.1),
\begin{equation} \label{Stokes_f}
\int_{\gamma(t)}\div_\Gamma\bu_T \, ds=  \int_{\partial \gamma(t)}\bu_T\cdot\nu \, d\ell;
\end{equation}
here $\nu=\nu(x,t)$ denotes the normal to $\partial \gamma(t)$ that is tangential to $\Gamma(t)$.

\medskip

\noindent{\bf Inextensibility.} We assume that the surface material is inextensible, i.e.  $\frac{d}{dt}\int_{\gamma(t)} 1 \, ds =0$ must hold.  The  Leibniz rule  yields $\int_{\gamma(t)} \divG \bu \, ds =0$. Since $\gamma(t)$ can be taken arbitrary, we get
\begin{equation} \label{massconser}
\divG \bu=0 \quad \text{on}~~\Gamma(t).
\end{equation}
We recall the notation $\kappa=\mbox{tr}(\bH) = \divG\bn$  for the (doubled) mean curvature. Equation \eqref{massconser} can be rewritten as
\begin{equation} \label{massconser2}
  \divG \bu_T= - u_N \kappa\quad \text{on}~~\Gamma(t).
\end{equation}
\medskip
\noindent{\bf Mass conservation.}  From $\frac{d}{dt}\int_{\gamma(t)} \rho(x,t) \, ds =0$, \eqref{Leibniz} and \eqref{massconser} we obtain $\dot \rho =0$. In particular, if $\rho|_{t=0}={\rm const}$, then $\rho={\rm const}$ for all $t>0$.
\\[1ex]
{\bf Momentum conservation.} The conservation of linear momentum for $\gamma(t)$ reads:
\begin{equation} \label{momconser}
 \frac{d}{dt} \int_{\gamma(t)} \rho \bu \, ds= \int_{\partial \gamma(t)} \mathbf{f}_{\nu} \, d\ell + \int_{\gamma(t)} \bb \, ds, % +\int_{\gamma(t)} F_{con} \bn \, ds,
\end{equation}
where $\mathbf{f}_{\nu}$ are the contact forces on $\partial \gamma(t)$, $\bb=\bb(x,t)$ are the area forces on $\gamma(t)$, which include both tangential and normal forces, for example, normal stresses induced by an ambient medium and elastic bending forces.
\\[1ex]
{\bf Surface diffusion.} For the modeling of the contact forces we use results from  \cite{GurtinMurdoch75,MurdochCohen79}. In \cite{GurtinMurdoch75}, Theorems  5.1 and 5.2, the ``Cauchy-relation'' $\mathbf{f}_\nu= \bT\nu$, with a symmetric tangential stress tensor $\bT$ is derived. We denote this \emph{surface stress tensor} by $\bsigma_\Gamma$, which has the properties $\bsigma_\Gamma=\bsigma_\Gamma^T$ and $\bsigma_\Gamma= \bP\bsigma_\Gamma \bP$.  In \cite{GurtinMurdoch75} the following (infinitesimal) \emph{surface rate-of-strain tensor} is derived:
\begin{equation} \label{strain}
 E_s(\bu):= \frac12 \bP (\nabla \bu +\nabla^T \bu)\bP = \frac12(\nabla_\Gamma \bu + \nabla_\Gamma^T \bu).
 \end{equation}
One needs a constitutive law which relates $\bsigma_\Gamma$ to this surface strain tensor. We consider a ``Newtonian surface fluid'', i.e., a constitutive law of the form
\[
  \bsigma_\Gamma= - \pi \bP + C( \gradG \bu),
\]
with a scalar function $\pi$, surface pressure, and a linear mapping $C$. Assuming  \textit{isotropy} and requiring an \textit{independence of the frame of reference} leads to the so-called \emph{Boussinesq--Scriven} surface stress tensor, which can be found at several places in the literature, e.g., \cite{aris2012vectors,BothePruess2010,GurtinMurdoch75,scriven1960dynamics} :
\begin{equation*} \label{BousS}
 \bsigma_\Gamma= - \pi \bP + (\lambda- \mu) (\divG \bu)\bP + 2 \mu E_s(\bu),
\end{equation*}
with an interface dilatational viscosity  $\lambda$ and interface shear viscosity $\mu >0$. We assume $\lambda$ and $\mu$ constant. Due to inextensibility the dilatational term vanishes, and we get
\begin{equation} \label{BousS1} \bsigma_\Gamma= - \pi \bP  + 2\mu E_s(\bu).
 \end{equation}
\smallskip

 Using the Stokes formula~\eqref{Stokes_f} applied row-wise to $\bsigma_\Gamma$ and the relation
$
  \divG (\pi \bP)=\nabla_\Gamma^T \pi - \pi \kappa \bn,
$
we obtain the following linear momentum balance for $\gamma(t)$:
\begin{equation*} \label{conservation}
  \frac{d}{dt} \int_{\gamma(t)} \rho \bu \, ds = \int_{\gamma(t)}( - \nabla_\Gamma^T \pi + 2 \mu \divG (E_s(\bu))  + \bb + \pi \kappa\bn  )\, ds.
\end{equation*}
For the left hand-side of this equation, the Leibniz rule~\eqref{Leibniz} gives
\begin{equation*}
 \frac{d}{dt} \int_{\gamma(t)} \rho \bu \, ds= \int_{\gamma(t)}( \dot \rho \bu + \rho \dot \bu + \rho \bu \divG \bu )\, ds.
\end{equation*}
The inextensibility and mass conservation yield the simplification
$
\dot \rho \bu + \rho \dot \bu + \rho \bu \divG \bu = \rho \dot \bu.
$
Hence, we finally obtain the \emph{surface Navier-Stokes equations}  for inextensible viscous material surfaces:
\begin{equation} \label{momentum}
\left\{
\begin{split}
 \rho \dot \bu & =- \nabla_\Gamma^T \pi + 2\mu \divG (E_s(\bu))  + \bb +  \pi \kappa\bn, \\
 \divG \bu  & =0.
\end{split}\right.
\end{equation}
Together with the equations $\dot{\rho}=0$ and $V_\Gamma=\bu\cdot \bn$, where $V_\Gamma$ is the normal velocity of $\Gamma$, and suitable initial conditions  this forms a closed system of six equations for six unknowns $\bu$, $p$, $\rho$, and $V_\Gamma$.

Clearly, the area forces $\bb$ coming from the adjacent inner and outer media are critical for the dynamics of the  material surface. For the example of an ideal bulk fluid, one may assume normal stresses due to the pressure drop between inner and outer phases, $\bb=\bn(p^{int}-p^{ext})$, where $p^{int}-p^{ext}$ may depend on the surface configuration, e.g., its interior volume. In an  equilibrium with $\bu=0$ this simplifies to the balance of the internal pressure and surface tension forces according to Laplace's law.  Such a balance will be more complex if there is only a shape equilibrium, i.e., $u_N=0$, but $\bu_T \neq 0$, cf. \eqref{reaction} below.  The area forces $\bb$ may also include  forces depending on the shape of the surface, such as those due to an  elastic bending energy (Willmore energy), cf. for example, \cite{Bonito2010,canham1970minimum,helfrich1973elastic}. These forces depend on geometric invariants and material parameters. Therefore $\bb$ may (implicitly) depend on $\bu$.

 Using a completely different approach the model \eqref{momentum} is also derived in \cite{Gigaetal}  and it is also found in \cite{BothePruess2010} (in this reference $\pi$ is treated as a constant parameter related to surface tension). Some variants of \eqref{momentum}  are used in  \cite{barrett2014stable,nitschke2012finite,lengeler2015stokes}, {  cf. the further discussion in seciton~\ref{s_other} below}. In \cite{barrett2014stable,lengeler2015stokes} the interface viscous fluid flow is  coupled with outer bulk fluids, and  for the velocity of the material surface  $\bu=:\bu_\Gamma$ one introduces the condition $\bu_\Gamma=(\bu^{bulk})_{|\Gamma}$, which means that both the normal and tangential components of surface and bulk velocities coincide. The condition for the tangential component corresponds to a ``no-slip'' condition at the interface. The condition  $\
bu_\Gamma=(\bu^{bulk})_{|\Gamma}$, allows to eliminate $\bu_\Gamma$ (using a momentum balance in a small bulk volume element
that contains the interface) and to deal with the surface forces (both  viscous and $\bb$) through a localized force term in the bulk Navier-Stokes equation. The surface pressure $\pi$ remains and is used to satisfy the inextensibility condition $\divG \bu=0$.    In \cite{nitschke2012finite} a special case of \eqref{momentum}, namely that of a stationary surface is considered, cf. \eqref{NSstat} below. %Based on   energy variation principles  the same model as in \eqref{momentum} is  derived and further studied in the recent report~\cite{Gigaetal}.

In certain cases, for example, when the inertia of the surface material dominates over the viscous forces in the bulk, it may be more appropriate to relax the no-slip condition $\bu_\Gamma=(\bu^{bulk})_{|\Gamma}$ and assume the coupling with the ambient medium only through the area forces $\bb$.
In such a situation the surface flow can not be ``eliminated'' and the system \eqref{momentum} becomes an important part of the surface--bulk fluid dynamics model. In Section~\ref{splitting} below we take a closer look at the normal and tangential dynamics defined  by \eqref{momentum}.
As far as we know, in the  literature the surface Navier-Stokes equations \eqref{momentum} on \emph{evolving} surfaces, without coupling to bulk fluids, have only been considered  in the recent paper \cite{Gigaetal}. Results of numerical simulations of such a model for  a \emph{stationary} surface, $u_N=0$, are presented in  \cite{nitschke2012finite}. This special case $u_N=0$ will be further addressed in section~\ref{splitting}.

\subsection{Directional splitting of the surface Navier-Stokes equations} \label{splitting}
Given the force term $\bb$, the system \eqref{momentum} determines $\bu=u_N \bn +\bu_T$ (and thus the evolution of $\Gamma(t)$), and there is a strong coupling between $u_N$ and $\bu_T$. There is, however, a clear distinction between the normal direction and the tangential direction (see, e.g., the difference in the viscous forces in normal and tangential direction in \eqref{Res1} and \eqref{Res1a}). In particular, the geometric evolution of $\Gamma(t)$  is completely determined by $u_N$ (which may depend on $\bu_T$). Therefore, it is of interest to split the equation \eqref{momentum} for $\bu$ into two coupled equations for $u_N$ and $\bu_T$.
We project the momentum  equation \eqref{momentum} onto the tangential space and normal space, respectively.

%Note that for $\bv \in H^1(\Gamma)^3$ we have $\dot \bv= \dot \bv^e$, which implies   $\dot \bv=\frac{\partial \bv}{\partial t}+(\nabla \bv)\bu= \frac{\partial \bv}{\partial t}+(\nabla \bv)\bu_T$.
First, we compute with the help of identities \eqref{id1}--\eqref{id1a}
\begin{equation}\label{aux2}
  \bP\dot \bu  =\dot{(\bP \bu)} -  \dot \bP \bu= \dot \bu_T - \dot \bP \bu = \dot \bu_T + (\dot \bn \cdot\bu_T)\bn+ u_N \dot \bn.
   \end{equation}
   Note that the last two terms on the right hand-side are orthogonal, since $\bn \cdot \dot \bn=0$.
Applying $\bP$ to both sides of \eqref{aux2} and using $\bP^2=\bP$ and $\bP\dot \bn=\dot \bn$, we also get
\begin{equation}\label{aux3}
  \bP\dot \bu  =  \partial^\bullet_\Gamma\bu_T +  u_N \dot \bn,
   \end{equation}
where $\partial^\bullet_\Gamma\bu_T: =  \bP\dot \bu_T$ can be interpreted as the covariant material derivative.
We also have
\begin{equation*}
 \bn \cdot \dot \bu  = \dot u_N - \dot \bn\cdot \bu = \dot u_N - \dot \bn \cdot \bu_T.
\end{equation*}
We thus get the following directional splitting of the equations in \eqref{momentum}:
\begin{equation}\label{NS}
\left\{\begin{aligned}
 \rho \dot \bu_T &= -\nabla_\Gamma^T \pi + 2\mu\bP \divG E_s(\bu) +  \bb_T - \rho\big((\dot \bn\cdot \bu_T)\bn +u_N \dot \bn \big), \\
 \rho \dot u_N & = 2 \mu \bn\cdot \divG E_s(\bu) +\pi \kappa + b_N + \rho \dot \bn\cdot \bu_T,\\
 \divG \bu_T   &= - u_N \kappa.
\end{aligned}\right.
\end{equation}
The material derivative of the tangential vector field on the left-hand side of the first equation in \eqref{NS}, in general,
is not tangential to $\Gamma(t)$. Its normal component is balanced  by the term $\rho(\dot \bn\cdot \bu_T)\bn$.
One can also write this equation only in tangential terms employing the identity \eqref{aux3} instead of \eqref{aux2}.
This results in the tangential momentum equation
\begin{equation} \label{tangent}
\rho \partial^\bullet_\Gamma\bu_T = -\nabla_\Gamma^T \pi + 2\mu\bP \divG E_s(\bu) +  \bb_T - \rho u_N \dot \bn.
\end{equation}
These equations can be further rewritten using
\begin{equation} \label{idfund}
E_s(\bu)=E_s(\bu_T) + u_N \bH.
\end{equation}
From this, the definition of the Bochner Laplacian and the relations in Lemma~\ref{LemGauss} we get
\begin{align*}
 \bP \divG E_s(\bu) & =\bP \divG E_s(\bu_T) + \bP \divG (u_N \bH) \\
  & = \frac12 \bP  \divG (\gradG \bu_T)+ \frac12 \bP  \divG (\gradG^T \bu_T) + u_N \bP\divG (\bH) + \bH \gradG^T u_N \\
& = \frac12 \Delta_\Gamma \bu_T +\frac12 K \bu_T + \frac12 \gradG^T \divG \bu_T  + u_N \gradG^T \kappa +\bH \gradG^T u_N.
\end{align*}
We would like to have a representation of $\bP \divG E_s(\bu)$ that does not include derivatives of $\bH$ or its invariants.
To this end, we note that  $\divG \bu_T   = - u_N \kappa$ implies
\[
 \gradG^T \divG \bu_T  + u_N \gradG^T \kappa= - \gradG^T(u_N \kappa)+ u_N \gradG^T \kappa= - \kappa \gradG^T u_N.
\]
Combining this we get
\begin{equation} \label{aux450}
 2 \mu \bP \divG E_s(\bu) = \mu\big(\Delta_\Gamma \bu_T+   K \bu_T- \gradG^T(\divG \bu_T) - 2(\kappa \bP-\bH)\gradG^T u_N\big).
\end{equation}
Note that $\kappa \bP-\bH$ has the same eigenvalues and eigenvectors as $\bH$, which follows from the relation $\kappa \bP-\bH= K \bH^\dagger$, cf.~\eqref{iddd}.
Thus, using \eqref{id1a}, \eqref{tangent}, and \eqref{aux450} we can rewrite \eqref{NS} as
 \begin{equation}\label{NSalt}
\left\{\begin{aligned}
 \rho \partial^\bullet_\Gamma\bu_T &= -\nabla_\Gamma^T \pi + \mu\big(\Delta_\Gamma\bu_T + K\bu_T - \gradG^T(\divG \bu_T) - 2(\kappa \bP-\bH)\gradG^T u_N\big) \\ & \quad + \bb_T - \rho u_N \dot \bn \\
 \rho \dot u_N & = - 2\mu (\tr(\bH \gradG \bu_T)+u_N \tr(\bH^2)) +\pi \kappa + b_N + \rho \dot \bn\cdot \bu_T\\
 \divG \bu_T   &= - u_N \kappa.
\end{aligned}\right.
\end{equation}
It is interesting to note that the first equation in \eqref{NSalt} is of (quasi-)parabolic type, while the equation for the evolution of the normal velocity involves only first order derivatives.
Furthermore, $\dot \bn$ can be expressed in terms of $\bu_T$ and $u_N$,
\[
 \dot{\bn}=\bH \bu_T - \gradG^T u_N .
\]
Hence the derivatives in the terms $\rho u_N \dot \bn$ and $\rho \dot \bn\cdot \bu_T$ on the right-hand side of \eqref{NS} and \eqref{NSalt} are only \emph{tangential} ones (no $\frac{\partial}{\partial t}$ involved). From this we conclude that given  $\bu(\cdot, t)$ for $t < t_{\ast} $ (which determines $\Gamma(t)$, $t <t_\ast $) the second equation in \eqref{NSalt} determines the dynamics of $u_N(\cdot,t)$ at $t=t_\ast$, hence of the surface $\Gamma(t_\ast)$, and the first equation \eqref{NSalt} determines the dynamics of $\bu_T(\cdot,t)$ at $t=t_\ast$.

\begin{rem}\label{remArroyo} \rm
We already noted that \eqref{momentum} or \eqref{NSalt} together with  $\dot{\rho}=0$ and the equation for the surface evolution, $V_\Gamma=\bu\cdot \bn$, form a closed system. This is different to the situation in \cite{Gigaetal}, where the evolution of $\Gamma(t)$ is given \textit{a priori}, resulting in an overdetermined system (for the total velocity $\bu$), which is then projected to obtain a closed system for the tangential velocity $\bu_T$, cf. the discussion in section 1 of \cite{Gigaetal}.
 In our setting  an \textit{a priori} known evolution of the surface would imply that $u_N$ is given. In this case, the first and the third equations in \eqref{NS} or \eqref{NSalt} define a closed system for $\bu_T$ and $\pi$. We note, however, that the continuum mechanics corresponding to such a closed system is less clear to us, since the fundamental momentum balance \eqref{momconser} used to derive the equations does not assume any \textit{a priori} constraint on $u_N$.

 The model \eqref{momentum}, or equivalently the one in \eqref{NSalt}, differs from the fluid model on evolving surfaces  derived in \cite{arroyo2009}. In the latter a tangential momentum equation (eq. (3) in \cite{arroyo2009}) is introduced, which is similar to, but different from, the first equation in \eqref{NS}. The model in \cite{arroyo2009} is based on a ``conservation of linear momentum \emph{tangentially} to the surface'', which is not precisely specified\footnote{Footnote added in proofs of the accepted paper: The controversy was recently addressed in  Reuther, S. \&  Voigt, A., Erratum: The Interplay of Curvature and Vortices in Flow on Curved Surfaces. \textit{Multiscale Modeling \& Simulation} 16 (2018), 1448--1453.}. Our model is derived based on a conservation of total momentum (i.e. for $\bu$, not for $\bu_T$) as in \eqref{momconser}.  The ``tangential'' equation (1.2)  in the paper~\cite{Gigaetal} is the same as the one obtained by applying the projection $\bP$ to the first equation in \eqref{momentum}. Above it is shown that this projected equation equals \eqref{tangent} and also the first equations in \eqref{NS}, \eqref{NSalt}.
\end{rem}

We next discuss two special cases.

Firstly, assume that the system evolves to an equilibrium with $\Gamma(t)$ stationary, i.e.,  $u_N=0$. Then the equations in \eqref{NS} reduce to the following surface incompressible Navier-Stokes equations for the tangential velocity $\bu_T$ on a \emph{stationary} surface $\Gamma$:
\begin{equation}\label{NSstat}
\left\{
\begin{aligned}
\rho \left( \frac{\partial \bu_T}{\partial t}+(\bu_T\cdot\nabla_\Gamma)\bu_T \right)&= -\nabla_\Gamma^T \pi + 2\mu\bP \divG E_s(\bu_T) + \bb_T  \\
 \divG \bu_T   &= 0 .
\end{aligned}\right.
\end{equation}
For the derivation of the first equation in \eqref{NSstat} we used the tangential momentum equation \eqref{tangent}, $u_N=0$,  and
\[ \begin{split}
  \partial^\bullet_\Gamma\bu_T & = \bP(\frac{\partial \bu_T}{\partial t} + (\nabla \bu_T) \bu)=\bP(\frac{\partial \bu_T}{\partial t} + (\nabla \bu_T) \bu_T) \\ & = \frac{\partial \bu_T}{\partial t} + (\nabla_\Gamma \bu_T) \bu_T =: \frac{\partial \bu_T}{\partial t}+(\bu_T\cdot\nabla_\Gamma)\bu_T,
\end{split} \]
where for the third identity we used
$\frac{\partial\bP}{\partial t}=0$ on geometrically steady surfaces and  $(\nabla_\Gamma \bu_T) \bu_T=\bP(\nabla \bu_T)\bP \bu_T=\bP(\nabla \bu_T)\bu_T$ by \eqref{defdeli2}.
The second equation in \eqref{NS}, or \eqref{NSalt}, reduces to
\begin{equation} \label{reaction}
 b_N= 2\mu \tr (\bH \gradG \bu_T) - \pi  \kappa -\rho  \bu_T\cdot \bH \bu_T,
\end{equation}
which describes the reaction force $b_N$ of the surface flow $\bu_T$. If there is no surface flow, i.e.,  $\bu_T=0$, this reaction force is the usual surface tension $\pi \kappa$, with a surface tension coefficient $\pi$.

 Again, if the \emph{stationary} surface $\Gamma$ is \textit{a priori} given as a domain where the equations are posed, then \eqref{NSstat} (with a suitable initial condition) forms a complete system.
Equation \eqref{reaction} applies to a \textit{material} surface and can be seen as a necessary condition for the area normal force $b_N$ to sustain the geometrical equilibrium of the surface.

In the second case, $\Gamma(0)$  is taken equal to the plane $z=0$ in $\Bbb{R}^3$. This is not a closed surface, but the derivation above also applies to connected surfaces without boundary, which may be unbounded. We consider $b_N=0$, $u_N(0)=1$. Only easily checks that independent of $\bu_T$ the second equation in \eqref{NS} is satisfied for $u_N(\cdot,t)=1$, $\dot{\bn}=0$, $\bH=0$ for all $t \geq 0$. Hence,  the evolving surface is given by the plane $\Gamma(t)=\{\, (x,y,z)=(x,y,t)\,\}$. The first and the third equations in \eqref{NS} reduce to the standard planar Navier-Stokes equations for $\bu_T$.

\subsection{Other formulations of  the surface Navier--Stokes equations} \label{s_other}
{  Incompressible Navier-Stokes equations on \emph{stationary} manifolds are well-known in the  literature, e.g.,~\cite{ebin1970groups,Temam88,cao1999navier,taylor1992analysis,mitrea2001navier}. Only very few papers treat incompressible Navier-Stokes equations on \emph{evolving} surfaces, cf.~\cite{arroyo2009,Gigaetal}.  Comparing the model \eqref{momentum} for the general case of an evolving surface (or its equivalent reformulations treated above) or the model~\eqref{NSstat}  for the case of a stationary surface to the models treated in the literature we observe the following. As noted above (Remark~\ref{remArroyo}), the general model \eqref{momentum} is the same as the one derived in \cite{Gigaetal}, but differs from the one given in \cite{arroyo2009}. Differences with other models presented in the literature can result from a different treatment of surface diffusion or of surface pressure. Below we briefly address these two modeling topics.}\\[1ex]
{\bf Surface diffusion.}
Our modeling of diffusion is based on the constitutive law \eqref{BousS1}, leading to the second order term $\divG (E_s(\bu))$ in \eqref{momentum}, or $\bP \divG (E_s(\bu_T))$ for the stationary case in \eqref{NSstat}.  Certain other Navier-Stokes equations in the literature are formally obtained by substituting Cartesian differential operators by their geometric counterparts~\cite{Temam88,cao1999navier} rather than from  first mechanical principles.
This leads to  formulations of surface Navier-Stokes equations which are not necessarily equivalent, due to a difference in the diffusion terms.
The diagram below and identities ~\eqref{LaplacianALL} illustrate some ``correspondences'' between Cartesian and surface operators, where for the surface velocities we assume $u_N=0$, i.e., $\bu=\bu_T$,
\[
\begin{array}{rcccccc}
\mathbb{R}^{n-1}\,: &-\div(\nabla\bu+\nabla^T\bu)& \overset{\div\bu=0}{=}  & -\Delta\bu& {=}  & (\mbox{rot}^T\mbox{rot}-\nabla\div)\bu\\
 &\wr& &\wr& &\wr\\
 \text{Manifold}\,: & \underbrace{-\bP \divG (2 E_s(\bu))}& \overset{\divG\bu=0}{\neq}  &  \underbrace{-\Delta_\Gamma\bu}& {\neq}  & \underbrace{-\Delta_\Gamma^H\bu}\\
 &\text{surface}& &\text{Bochner }& &\text{Hodge}\\
 &\text{diffusion}& &\text{Laplacian}& &\text{Laplacian}\\
\end{array}
\]
Moreover, for a surface in $\mathbb{R}^3$ we have, cf. \eqref{Laplacian}, \eqref{Res2} and the Weitzenb\"{o}ck identity~\cite{rosenberg1997laplacian}, the following equalities for $\bu$ such that $\divG\bu=0$:
\begin{equation}\label{LaplacianALL}
-\bP \divG (2E_s(\bu))=-\Delta_\Gamma\bu-K\bu=-\Delta_\Gamma^H\bu-2K\bu.
\end{equation}
Using this we see that the Navier-Stokes system \eqref{NSstat}  coincides with the Navier-Stokes equations (on a stationary surface) considered in \cite{taylor1992analysis,mitrea2001navier} (see \cite{taylor1992analysis} section 6).
Formulations of the surface momentum equations employing the identity \[-\bP \divG (2E_s(\bu))=-\Delta_\Gamma^H\bu-2K\bu,\]
%Using \eqref{LaplacianALL} we obtain the following reformulation of \eqref{momentumtangential}:
%\begin{equation} \label{Stokesstrong}
% \begin{split}
%  \mu \Delta_\Gamma^H \bu_T + 2\mu K \bu_T +\nabla_\Gamma \pi =  \tilde \bb,\\
%  \divG \bu_T & =0.
% \end{split}
% \end{equation}
with the  Hodge--de Rham Laplacian $-\Delta_\Gamma^H$ can be convenient for rewriting the problem in  surface stream-function -- vorticity variables, see, e.g., \cite{nitschke2012finite}. However, such a formulation is less convenient for the analysis of  well-posedness, since the Gauss curvature $K$ in general does not have a fixed sign.
Moreover,  in a numerical approximation of \eqref{NSstat} one would have to approximate the Gauss curvature $K$ based on a ``discrete'' (e.g., piecewise planar) surface  approximation, which is known to be a delicate numerical issue.

\noindent
%In the remainder we restrict our discussion to   the formulation with
%the surface rate-of-strain tensor $\bP \divG (E_s(\bu))$.
%Using \eqref{LaplacianALL} we see that the Navier-Stokes system \eqref{NSstat}, which is the special case of \eqref{momentum} for a stationary surface, coincides with the Navier-Stokes %equations (on a stationary surface) considered in \cite{taylor1992analysis,mitrea2001navier} (see \cite{taylor1992analysis} section 6).
%We  note that the authors of \cite{Gigaetal} also considered equations \eqref{momentum} simplified to the case of a stationary surface. They, however,  claim to obtain a system different  from the one in  \cite{taylor1992analysis} {\color{blue}(the cause of confusion  is explained in \cite{miura2017singular}, where the author deduced tangential equations \eqref{tangent} and \eqref{NSstat} from the thin-domain limit of the regular three-dimensional Navier--Stokes problem).
{  {\bf Surface pressure.}  We discuss the derivation of the pressure terms $\nabla_\Gamma^T \pi$ and $\pi \kappa\bn$ in  \eqref{momentum}. In most other papers on surface Navier-Stokes equations a pressure term of the form  $\nabla_\Gamma^T \pi$ appears. In many papers, e.g., \cite{arnol2013mathematical,ebin1970groups,Temam88}, the term  $\pi \kappa\bn$ does not appear. We comment on this.
  The term $\pi \kappa\bn$ is part of the tension force generated by the fluidic surface and is due to the \textit{material} nature of the surface itself. The constitutive law \eqref{BousS1} and the momentum conservation yield  the term $\divG (\pi \bP)= \gradG^T\pi- \pi \kappa \bn$, which contains both tangential ($\gradG^T\pi$) and normal ($\pi \kappa \bn$) forces.

While the  present paper introduces surface pressure via the surface stress tensor $\bsigma_\Gamma$, as is common in continuum mechanics,   one can use a Hodge type decomposition to introduce $\pi$  as a Lagrange  multiplier corresponding to the divergence constraint, see, e.g.~\cite{ebin1970groups}. Also in this setting one obtains the term $\pi \kappa \bn$ if one considers general (\emph{not} necessarily tangential) vector fields on the surface.  This can be seen as follows. The following result can be  proved (see, Lemma~2.7 in \cite{Gigaetal}) for a smooth   surface $\Gamma$: For $\bu\in L^2(\Gamma)^3$
\[
\int_\Gamma \bu\cdot\bphi\,ds=0\quad\forall\bphi\in C^1_0(\Gamma)^3,~\div_\Gamma\bphi=0\qquad\text{iff}\qquad\bu=\nabla_\Gamma^T\pi-\pi\kappa\bn,
\]
for a $\pi\in H^1(\Gamma)$. For closed surfaces $C^1_0(\Gamma)^3$ can be replaced by $C^1(\Gamma)^3$. Hence,
the $L^2$-orthogonal complement to the space of smooth solenoidal vector functions on $\Gamma$ leads to  pressure terms exactly the same as in \eqref{momentum}.

If on the other hand, one only considers \emph{tangential} vector fields (which is natural for stationary surfaces) the derivation of the above result in  \cite{Gigaetal} also yields the following, which is a surface variant of a well-known Helmholtz type result (cf. Theorem 2.9 in \cite{GR}): For $\bu\in L^2(\Gamma)^3$ such that $\bu =\bu_T$ holds, we have
\[
\int_\Gamma \bu_T \cdot\bphi_T \,ds=0\quad\forall\bphi\in C^1_0(\Gamma)^3,~\div_\Gamma\bphi_T=0\qquad\text{iff}\qquad\bu_T=\nabla_\Gamma^T\pi,
\]
for a $\pi\in H^1(\Gamma)$. Hence, if the pressure is considered as a Lagrange multiplier (for the divergence free constraint)  the term  $\pi \kappa \bn$ occurs if  nontangential
velocity fields are present (as in the case of evolving surfaces).}

\subsection{Surface Stokes problem} The mathematical analysis of well-posedness of a problem as in \eqref{NS} (or \eqref{momentum}) is a largely open question. In this paper, we study the well-posedness of a relatively  simple special case, namely a \emph{Stokes} problem on a \emph{stationary} surface. We assume that $u_N=0$ (stationary surface) and assume that the viscous surface forces dominate and thus it is reasonable to skip  the nonlinear $\bu_T \cdot \gradG \bu_T$ term in the  material derivative. Furthermore, we first restrict to the equilibrium flow problem, i.e., $\frac{\partial \bu_T}{\partial t}=0$.
%and comment on the time dependent Stokes problem in Remark~\ref{remtdependent}.
We thus obtain  the \emph{stationary surface Stokes} problem
 \begin{equation} \label{Stokes1} \begin{split}
  - 2\mu \bP \divG (E_s(\bu_T)) +\nabla_\Gamma^T \pi &=  \bb_T,\\
  \divG \bu_T & =0.
 \end{split}
 \end{equation}
 One readily observes that all constant pressure fields and tangentially rigid surface fluid motions,  i.e., motions satisfying  $E_s(\bv_T)=0$,
 are in the kernel of the differential operator on the left-hand side of the equation. Integration by parts,
 immediately implies  the necessary
 consistency condition for the right-hand side of \eqref{Stokes1},
 \begin{equation}\label{constr}
% \begin{split}
 \int_\Gamma \bb_T \bv_T\,ds=0\quad\text{for all}~~\bv_T~~ \text{s.t.}~~ E_s(\bv_T)=0.
% \end{split}
 \end{equation}
In the following sections we analyze different weak formulations of this Stokes problem.

The subspace of all tangential vector fields $\bv_T$ on $\Gamma$ satisfying  $E_s(\bv_T)=0$ plays an important role in the
analysis of the surface Stokes problem. In the literature, such fields are known as \textit{Killing vector fields}, see, e.g., \cite{sakai1996riemannian}. For a smooth two-dimensional Riemannian manifold, Killing vector fields form a Lie algebra, which dimension
is at most 3. For a compact smooth surface  $\Gamma$ embedded in $\mathbb{R}^3$ the dimension of the algebra is 3 iff $\Gamma$ is isometric to a sphere.

 \section{A well-posed  variational surface Stokes equation}\label{s_Stokes}
 Assume that $\Gamma$ is a closed sufficiently smooth manifold.
 We introduce the space $V:=H^1(\Gamma)^n$, with norm
 \begin{equation} \label{H1norm}
  \|\bu\|_{1}^2:=\int_{\Gamma}\|\bu(s)\|_2^2 + \|\nabla\bu^e (s)\|_2^2\,ds,
 \end{equation}
where $\|\cdot\|_2$ denotes the vector and matrix $2$-norm. Here $\bu^e$ denotes the constant extension along normals of $\bu:\Gamma \to \Bbb{R}^n$. We have $\nabla\bu^e= \nabla (\bu\circ p)=\nabla\bu^e\bP$, where $p$ is the closest point projection onto $\Gamma$, hence only {tangential} derivatives are included in this $H^1$-norm.
%Further  for $\bu^e$ we shall use the notation $\bu$, assuming  always the constant extensions along normals if ambient derivatives appear in formulas for $\bu$ defined on $\Gamma$.
 We define the spaces
\begin{equation}   \label{defVT}
 V_T:= \{\, \bu \in V~|~ \bu\cdot \bn =0\,\},\quad E:= \{\, \bu \in V_T~|~ E_s(\bu)=0\,\}.
\end{equation}
Note that $E$ is a closed subspace of $V_T$ and $\mbox{dim}(E)\le 3$.
  We use an orthogonal decomposition $V_T=V_T^0 \oplus E$ with the Hilbert space $V_T^0 = E^{\perp_{\|\cdot\|_{1}}}$ (hence $V_T^0 \sim V_T/E$). We also need the factor space
 $L_0^2(\Gamma):=\{\, p \in L^2(\Gamma)~|~ \int_\Gamma p\,dx=0\,\}\sim L^2(\Gamma)/\mathbb{R}$.
We introduce the bilinear forms
\begin{align}
a(\bu,\bv)& := 2\mu \int_\Gamma E_s(\bu):E_s(\bv) \, ds= 2\mu \int_\Gamma {\rm tr}\big(E_s(\bu) E_s(\bv)\big) \, ds, \quad \bu,\bv \in V, \label{defblfa} \\
b(\bu,p) &:= - \int_\Gamma p\,\divG \bu \, ds,  \quad \bu \in V, ~p \in L^2(\Gamma). \label{defblfb}
\end{align}
%(Note that in the definition of the strain tensor $E_s(\bu)$ as in \eqref{strain} we need the constant extension $\bu^e$ of $\bu \in V$, similarly for $\divG \bu$).
We take $f \in  V'$, such that $f(\bv_T)=0$ for all $\bv_T\in E$, and consider the following variational Stokes problem: determine $(\bu_T,p) \in V_T^0 \times L_0^2(\Gamma)$ such that
 \begin{equation} \label{Stokesweak1} \begin{split}
           a(\bu_T,\bv_T) +b(\bv_T,p) &=f(\bv_T) \quad \text{for all}~~\bv_T \in V_T, \\
           b(\bu_T,q) & = 0 \qquad \text{for all}~~q \in L^2(\Gamma).
                                      \end{split}
 \end{equation}
This weak formulation is consistent to the strong one in \eqref{Stokes1} for $f(\bv_T)=(\mathbf{b}_T,\bv_T)_0$. Note that $ E_s({\bv_T})=0$ implies $\tr(\gradG \bv_T)=0$ and thus $\divG {\bv_T}=0$, hence, $b(\bv_T,p)=0$ for all $\bv_T \in E$. From this it follows that the first equation in  \eqref{Stokesweak1} is always satisfied for all $\bv_T \in E$, hence it is not relevant whether we use $V_T$ or $V_T^0$ as space of test functions.
For the analysis of well-posedness a surface  Korn's inequality is a crucial ingredient. Although there are results in the literature on Korn's type equalities on surfaces, e.g. \cite{CiarletElasticity,Mardare}, these are related to surface models of thin shells, such as Koiter's model, which contain derivatives in the direction of the normal displacement. In the literature we did not find a result of the type given in \eqref{korn} below, and therefore we include a proof.
\begin{lem} \label{Kornlemma}
Assume $\Gamma$ is $C^2$ smooth {   and compact}. There exists $c_K >0$ such that
 \begin{equation} \label{korn}
 \|E_s(\bu)\|_{L^2(\Gamma)} \geq c_K \|\bu\|_{1} \quad \text{for all}~~\bu \in V_T^0.
 \end{equation}
\end{lem}
\begin{proof}
Let $\bu= \bu_T \in V_T^0$ be given. Throughout this proof, the extension $\bu^e$ is also denoted by $\bu$.
Since $\nabla\bu^e=\nabla\bu$ includes only tangential derivatives we introduce the notation
\[ \nabla_P\bu:=(\nabla\bu)\bP= \nabla\bu^e\]
for the tangential derivative. Furthermore,
the symmetric part of the tangential derivative tensor is denoted by $
 \textbf{e}_s(\bu):=\frac12( \nabla_P\bu  + \nabla_P^T \bu)$. Below we derive the following inequality:
\begin{equation} \label{Korn1}
 \|\bu\|_{L^2(\Gamma)}+\|\textbf{e}_s(\bu)\|_{L^2(\Gamma)} \geq c \|\bu\|_{1} \quad \text{for all}~~\bu \in V_T.
 \end{equation}
%Denote by $\bH$ the second fundamental tensor on $\Gamma$, $\bH=\nabla  \bn$.
%From $\bu\cdot \bn=0$ it follows that  $\bn^T (\nabla  \bu)  +  \bu^T (\nabla  \bn) =0$ holds, and thus $\nabla^T\bu  \bn= - \bH \bu$. Using this we get
Recall \eqref{aux1}, $\bH\bu = - (\nabla^T \bu) \bn$. Using this and $\bP=\bI-\bn\bn^T$ we get
$\gradG^T \bu= \bP \nabla^T\bu \bP=  \bP \nabla^T\bu - \bP (\nabla^T\bu) \bn \bn^T = \nabla_P^T\bu + \bH \bu \bn^T$, and thus we get the identity
\[
E_s(\bu)=\textbf{e}_s(\bu)+\frac12\left(\bH\,\bu \bn^T+\bn \bu^T \bH \right).
\]
Since the surface is $C^2$-smooth  this equality implies $\|\textbf{e}_s(\bu)\|_{L^2(\Gamma)}\leq \|E_s(\bu)\|_{L^2(\Gamma)} + c \|\bu\|_{L^2(\Gamma)}$,  and combining this with \eqref{Korn1} yields
 \begin{equation} \label{Korn2}
 \|\bu\|_{L^2(\Gamma)}+\|E_s(\bu)\|_{L^2(\Gamma)} \geq c \|\bu\|_{1} \quad \text{for all}~~\bu \in V_T,
 \end{equation}
with some $c>0$.
We now apply the Petree-Tartar Lemma, e.g. Lemma A.38 in \cite{Ern04} to $E_s \in \mathcal{L}(V_T^0, L^2(\Gamma)^{3 \times 3})$, which is injective, and the compact embedding ${\rm id}: V_T^0 \to L^2(\Gamma)^3$.  Application of this lemma yields the desired result.

It remains to proof the inequality \eqref{Korn1}.  We use a local parametrization of $\Gamma$ and a standard Korn's inequality in Euclidean space.

%{\bf AR: maybe the analysis below has to be shortened; essentially the same as in Lecture notes  of Delfour-Zolesio  , but more precise}\\
Let $\omega \subset \Bbb{R}^{n-1}$ be a bounded open connected domain and $\Phi: \omega \to \Gamma$ a local parametrization of $\Gamma$; $\{\bxi_1,\dots,\bxi_{n-1}\}$ denotes the Cartesian basis in  $\Bbb{R}^{n-1}$. Partial derivatives of $\Phi(\xi)=\Phi (\xi_1, \ldots, \xi_{n-1})$ are denoted by $\ba_\alpha(\xi):= \frac{\partial \Phi(\xi)}{\partial \xi_\alpha} \in \Bbb{R}^n$, $\alpha=1, \ldots, n-1$. Below we often skip the argument $\xi \in \omega$. Greek indices always range from $1$ to $n-1$, and roman indices from $1$ to $n$. We furthermore define $\ba_n:=\bn$. The dual basis (or contravariant basis) is given by $\ba^\beta$ such that $\bP \ba^\beta =\ba^\beta$ and  $\ba^\beta \cdot \ba_\alpha=0$  for $\alpha \neq \beta$ and $\ba^\beta \cdot \ba_\beta =1$. Furthermore $\ba^n:=\ba_n$. Note that $\bP \ba_\alpha= \ba_\alpha,~\bP\ba^\alpha=\ba^\alpha$, $\bP\ba_n=\bP\ba^n=0$.
A given vector function $\bu:\, \Gamma \to \Bbb{R}^n$ is pulled back to $\omega$ as follows:
\[
 \vec{\bu}= (\vec{u}_1, \ldots, \vec{u}_{n-1}):\, \omega \to \Bbb{R}^{n-1}, \quad \vec{u}_\alpha:= (\bu \circ \Phi)\cdot \ba_\alpha.
\]
Note that $\bu\circ \Phi= \vec{u}_\alpha \ba^\alpha$ (Einstein summation convention). We also use the standard notation $\vec{u}_{\alpha,\beta}:= \frac{\partial \vec{u}_\alpha}{\partial \xi_\beta}$. Note that $(\ba^\lambda \cdot \ba_\alpha)_{,\beta}=0$ and thus $\ba^\lambda \cdot \ba_{\alpha,\beta}=- \ba^\lambda_{,\beta}\cdot \ba_\alpha$ holds. Using this we get
\begin{align*}
 \vec{u}_{\alpha,\beta} &= \ba_\alpha \cdot \nabla (\bu\circ \Phi) \bxi_\beta +(\bu\circ \Phi)\cdot \ba_{\alpha,\beta}
 =  \ba_\alpha \cdot (\nabla\bu\circ \Phi ) \ba_\beta + (\vec{u}_\lambda \ba^\lambda) \cdot \ba_{\alpha,\beta}\\
& = \ba_\alpha \cdot (\nabla_P\bu\circ \Phi ) \ba_\beta + \vec{u}_\lambda (\ba^\lambda \cdot \ba_{\alpha,\beta}) = \ba_\alpha \cdot (\nabla_P\bu\circ \Phi ) \ba_\beta - \vec{u}_\lambda \ba^{\lambda}_{,\beta} \cdot \ba_\alpha.
 \end{align*}
 Now note that for $\xi \in \omega$ and $x:=\Phi(\xi)$ we have
 \[ \begin{split}
  \nabla_P(\ba^\lambda\circ \Phi^{-1}(x))\ba_\beta(\xi) & = \nabla (\ba^\lambda\circ \Phi^{-1}(x))\ba_\beta(\xi) =
  \nabla\ba^\lambda(\xi)\nabla\Phi^{-1}(x)\ba_\beta(\xi) \\ & = \nabla\ba^\lambda(\xi)\left[\nabla\Phi(\xi)\right]^{-1}\ba_\beta(\xi) =
   \nabla\ba^\lambda(\xi) \bxi_\beta =\frac{\partial \ba^\lambda (\xi)}{\partial \xi_\beta} = \ba^{\lambda}_{,\beta}(\xi).
 \end{split} \]
Using this in the relation above we obtain
\begin{equation} \label{rela}
 \vec{u}_{\alpha,\beta}(\xi)= \ba_\alpha(\xi) \cdot \big( \nabla_P\bu(x)- \vec{u}_\lambda(\xi) \nabla_P(\ba^\lambda \circ\Phi^{-1})(x)\big) \ba_\beta(\xi), \quad \xi \in \omega, ~x=\Phi(\xi).
\end{equation}
The symmetric part of the Jacobian in $\Bbb{R}^{n-1}$ is denoted by $E(\vec{\bu})_{\alpha \beta}= \frac12\big(\vec{u}_{\alpha,\beta}+ \vec{u}_{\beta,\alpha}\big)$. Thus we get (we skip the arguments again):
\begin{equation} \label{rela1}
 E(\vec{\bu})_{\alpha \beta}=  \ba_\alpha \cdot \big(\textbf{e}_s(\bu) - \vec{u}_\lambda \textbf{e}_s(\ba^\lambda \circ\Phi^{-1})\big)\ba_\beta.
\end{equation}
From this we get, using the $C^2$ smoothness of the manifold:
\begin{equation} \label{rela2}
\|E(\vec{\bu})(\xi)\|_{2} \leq c(\|\textbf{e}_s(\bu)(x)\|_{2} + \|\vec{\bu}(\xi)\|_{2}) \leq c(\|\textbf{e}_s(\bu)(x)\|_{2} + \|\bu(x)\|_{2}),
\end{equation}
for $\xi \in \omega, ~x= \Phi(\xi)$.
Now we derive a bound for  $\|\nabla_P\bu(x)\|_{2}$ in terms of $\|\nabla\vec{\bu}(\xi)\|_{2}$. Let $e_i$ be the standard basis in $\Bbb{R}^n$. Note that $e_i=(e_i\cdot \ba^l)\ba_l$. Using this, $(\nabla_P\bu)\bn=0$ and \eqref{rela}  we get (we skip the arguments $\xi$ and $x$):
\begin{align*}
 e_j \cdot \nabla_P\bu e_i & = (e_i \cdot  \ba^l)(e_j \cdot \ba^m) \ba_m \nabla_P\bu\, \ba_l  \\ & =
   (e_i \cdot  \ba^\beta)(e_j \cdot \ba^\alpha) \ba_\alpha\cdot \nabla_P\bu\, \ba_\beta + (e_i \cdot  \ba^\beta)(e_j \cdot \bn) \bn\cdot \nabla_P \bu \, \ba_\beta  \\
   &= (e_i \cdot  \ba^\beta)(e_j \cdot \ba^\alpha)\Big( \vec{u}_{a,\beta} +\vec{u}_\lambda \ba_\alpha \cdot \nabla_P(\ba^\lambda \circ\Phi^{-1}) \ba_\beta\Big) \\ &\quad  +(e_i \cdot  \ba^\beta)(e_j \cdot \bn) \bn\cdot \nabla_P \bu \, \ba_\beta.
\end{align*}
Note that
\[
 \bn\cdot \nabla_P \bu \, \ba_\beta= \bn\cdot (\nabla  \bu) \bP \ba_\beta= \bn\cdot (\nabla  \bu) \ba_\beta=
  ( \nabla  \bu)^T \bn
  \cdot \ba_\beta= - \bH \bu\cdot \ba_\beta = - \bu \cdot \bH \ba_\beta.
\]
Using this in the relation above and using the smoothness of $\Gamma$ then yields
\begin{equation} \label{rela3}
 \| \nabla_P\bu(x)\|_{2} \leq c \big (\|\nabla\vec{\bu}(\xi)\|_{2} +\|\vec{\bu}(\xi)\|_{2}+ \|\bu(x)\|_{2}) \leq  c \big (\|\nabla\vec{\bu}(\xi)\|_{2} +\|\vec{\bu}(\xi)\|_{2}),
\end{equation}
for $\xi \in \omega$, $x=\Phi(\xi)$. For $\omega \subset \Bbb{R}^{n-1}$ we have the Korn inequality
\begin{equation} \label{K1}
 \int_{\omega}( \|E(\vec{\bu})\|_{2}^2 + \|\vec{\bu}\|_{2}^2 )\, d\xi \geq c_K \int_\omega \|\nabla\vec{\bu}\|_{2}^2 \, d\xi ,
\end{equation}
with $c_K = c_K(\omega)>0$. Since $\Gamma$ is compact, there is a finite number of maps $\Phi_i: \omega_i \to \Phi_i(\omega_i) \subset \Gamma$, $i=1,\ldots, N$, which form a parametrization of $\Gamma$. Using the results in \eqref{rela3}, \eqref{K1} and  \eqref{rela2} we then get
\begin{align*}
 \|\bu\|_{1}^2 & = \int_{\Gamma} \|\nabla_P\bu(x)\|_{2}^2 +\|\bu(x)\|_{2}^2 \, dx \leq  N \max_{1\leq i \leq N} \int_{\Phi_i(\omega_i)}\|\nabla_P\bu(x)\|_{2}^2 +\|\bu(x)\|_{2}^2 \, dx \\
  & \leq c  \int_{\omega_i} (\|\nabla\vec{\bu}(\xi)\|_{2}^2 +\|\vec{\bu}(\xi)\|_{2}^2)|\det (\nabla\Phi_i(\xi))| \, d\xi\\
  & \leq c \int_{\omega_i} \|E(\vec{\bu})(\xi)\|_{2}^2+ \|\vec{\bu}(\xi)\|_{2}^2 |\det (\nabla\Phi_i(\xi))| \,d\xi\\
  & \leq c \int_{\Phi_i(\omega_i)} \|\textbf{e}_s(\bu)(x)\|_{2}^2 + \|\bu(x)\|_{2}^2\, dx
   \leq c \int_{\Gamma} \|\textbf{e}_s(\bu)(x)\|_{2}^2 + \|\bu(x)\|_{2}^2\, dx,
\end{align*}
from which the inequality in \eqref{Korn1}  easily follows.
\end{proof}
\ \\[1ex]
Korn's inequality implies ellipticity of the bilinear form $a(\cdot,\cdot)$ on $V_T^0$. In the next lemma we treat the second main ingredient needed for well-posedness of the Stokes saddle point problem, namely an inf-sup property of $b(\cdot,\cdot)$.
\begin{lem} \label{leminfsup} Assume $\Gamma$ is $C^2$ smooth and closed.  The following inf-sup estimate holds:
\begin{equation} \label{infsup}
 \inf_{p\in L^2_0(\Gamma)}\sup_{\bv_T\in{V_T^0}}\frac{b(\bv_T,p)}{\|\bv_T\|_{1}\|p\|_{L^2}}  \geq c >0.
\end{equation}
\end{lem}
\begin{proof}
 Take $p \in L_0^2(\Gamma)$. Let $\phi \in H^1(\Gamma) \cap L_0^2(\Gamma)$ be the solution of
 \[
  \Delta_\Gamma \phi = p \quad \text{on}~~\Gamma.
 \]
For $\phi$ we have the regularity estimate $\|\phi\|_{H^2(\Gamma)} \leq c \|p\|_{L^2}$, with a constant $c$ independent of $p$ (see, e.g., Theorem 3.3 in \cite{DEreview}). Take $\bv_T:=-\nabla_\Gamma^T \phi \in V_T$, and the orthogonal decomposition $\bv_T= \bv_T^0 +\tilde \bv$, with $\bv_T^0 \in V_T^0$, $\tilde \bv \in E$.  We have $ \|\bv_T^0\|_{1}  \leq \|\bv_T\|_{1} \leq c \|\phi\|_{H^2(\Gamma)} \leq c\|p\|_{L^2}$. Furthermore, $ E_s(\tilde{\bv})=0$ implies $\div_\Gamma\tilde{\bv}=0$ and thus $b(\bv_T^0,p)=b(\bv_T,p)$. Using this we get
\begin{equation} \label{infsup1}
\frac{b(\bv_T^0,p)}{\|\bv_T^0\|_{1}}= \frac{b(\bv_T,p)}{\|\bv_T^0\|_{1}} = \frac{\int_{\Gamma} \Delta_\Gamma \phi \, p\, ds}{\|\bv_T^0\|_{1}} =\frac{\|p\|_{L^2}^2}{\|\bv_T^0\|_{1}} \geq c \|p\|_{L^2},
\end{equation}
which completes the proof.
\end{proof}
\medskip

\begin{thm}\label{Th1}
Assume $\Gamma$ is $C^2$ smooth and closed. The weak formulation \eqref{Stokesweak1} is well-posed.
\end{thm}
\begin{proof}
 Note that $\|E_s(\bu)\|_{L^2} \leq  \|\nabla\bu^e\|_{L^2}$ and $\|\divG \bu\|_{L^2}\leq n \|\gradG \bu\|_{L^2} = n \|\nabla\bu^e\|_{L^2}$ hold. From this it follows that the bilinear forms $a(\cdot,\cdot)$ and $b(\cdot,\cdot)$ are continuous on $V_T \times V_T$ and $V_T \times L_0^2(\Gamma)$, respectively. Ellipticity of $a(\cdot,\cdot)$ follows from Lemma~\ref{Kornlemma} and the inf-sup property of $b(\cdot,\cdot)$ is derived in Lemma~\ref{leminfsup}.
\end{proof}
%%f
%\ \\[1ex]
%\begin{rem} \label{remtdependent} \rm {\bf AR: this is a preliminary comment}
% We comment on the extension to time dependent Stokes, i.e., $\frac{\partial \bu_T}{\partial t}$ is included in \eqref{Stokes1}. I expect that the well-posed weak formulation \eqref{Stokesweak1} has a straightforward extension to a time-dependent quasi-parabolic problem using Bochner spaces (as in Ern-Guermond).....
%\end{rem}

\section{A well-posed variational Stokes problem with Lagrange multiplier} \label{SecLagrange}
In the formulation \eqref{Stokesweak1} the velocity $\bu_T$ is tangential to the surface. For Galerkin discretization methods, such as a finite element method, this may be less convenient, cf. Remark~\ref{rem2a}. In this section we consider a variational formulation in a space, which does not contain the constraint $\bn\cdot \bu =0$. The latter  is treated using a Lagrange multiplier.

We recall the notation $\bu = \bu_T + u_N\bn$ for $\bu \in V$ and we define the following
Hilbert space:
\[
V_\ast  :=\{\, \bu \in L^2(\Gamma)^n\,:\,\bu_T \in V_T,~u_N\in L^2(\Gamma)\,\},  \quad\text{with}~~
\|\bu\|_{V_\ast}^2:=\|\bu_T\|_{1}^2+\|u_N\|_{L^2(\Gamma)}^2.
\]
Note that $V_\ast\sim V_T \oplus L^2(\Gamma)$ and $E\subset V_T\subset V_\ast$ is a closed subspace of  $V_\ast$. Thus
the space  $V_\ast^0:=E^{\perp_{V_\ast}} \sim V_T^0 \oplus L^2(\Gamma)$ is a Hilbert space.
We introduce the bilinear form
\[
 \tilde  b(\bv,\{p,\lambda\})= - \int_{\Gamma} \divG \bv_T \, p \, ds +\int_{\Gamma}\lambda v_N \, ds
  = b(\bv_T,p) + (\lambda, v_N )_{L^2(\Gamma)}.
\]
on $V_\ast\times \left(L_0^2(\Gamma) \times L^2(\Gamma)\right)$.
Based on the identity \eqref{idfund} we introduce (with an abuse of notation, cf.~\eqref{defblfa}) the bilinear form
\begin{equation} \label{defaalt}
a(\bu,\bv) := 2\mu \int_\Gamma {\rm tr}\big((E_s(\bu_T)+u_N\bH)( E_s(\bv_T)+ v_N \bH)\big)\, ds, \quad \bu,\bv \in V_\ast.
\end{equation}

In this bilinear form we need $H^1(\Gamma)$ smoothness of the tangential component $\bu_T$ and only $L^2(\Gamma)$ smoothness of the normal component $u_N$. If the latter component has also $H^1(\Gamma)$ smoothness, then from \eqref{idfund} we get
\begin{equation} \label{eqa}
 a(\bu,\bv)= 2\mu \int_\Gamma {\rm tr}\big(E_s(\bu) E_s(\bv)\big)\, ds, \quad \text{for}~~\bu,\bv \in V.
\end{equation}
The bilinear form $a(\cdot,\cdot)$ is continuous:
\[
a(\bu,\bv)\le c\|\bu\|_{V_\ast}\|\bv\|_{V_\ast}\quad\forall~\bu,\bv\in V_\ast.
\]

For $f \in V_\ast'$ such that $f(\bv_T)=0$ for all $\bv_T\in E$, we consider the modified Stokes weak formulation: Determine $(\bu,\{p, \lambda\}) \in V_\ast^0 \times  \left(L_0^2(\Gamma) \times L^2(\Gamma)\right)$ such that
 \begin{equation} \label{Stokesweak2} \begin{split}
           a(\bu,\bv) +\tilde b(\bv,\{p,\lambda\}) &=f(\bv) ~ \text{for all}~~\bv \in V_\ast^0, \\
           \tilde b(\bu,\{q,\nu\}) & = 0 ~\text{for all}~~\{q,\nu\} \in L_0^2(\Gamma)\times L^2(\Gamma).
                                      \end{split}
 \end{equation}
%One easily checks that this weak formulation is consistent to the strong one in \eqref{Stokes1}.
%{\color{blue} Let the pair  $\bu_T$, $\pi$ be the velocity and pressure solution to the strong formulation \eqref{Stokes1}. One readily verifies that the triple $\bu=\bu_T$ , $p=\pi$, $\lambda=-{\rm tr}\big(E_s(\bu_T)\bH)$ satisfies \eqref{Stokesweak2} with $f(\bv)=(\bb_T,\bv)$. In this sense, \eqref{Stokesweak2} is another consistent weak formulation of the surface Stokes problem \eqref{Stokes1}.
{   This is a consistent weak formulation of the surface Stokes problem \eqref{Stokes1}, cf. Remark~\ref{lambazero} below. In that remark we also explain that the test space $V_\ast^0$ in the first equation in \eqref{Stokesweak2} can be replaced by $V_\ast$}.
\begin{thm} \label{ThmStokes2}
 The problem \eqref{Stokesweak2} is well-posed. Furthermore, its unique solution satisfies $\bu \cdot \bn=0$.
\end{thm}
\begin{proof}
 The bilinear forms $a(\cdot,\cdot)$ and $\tilde b(\cdot,\{\cdot,\cdot\})$ are continuous on $V_\ast \times V_\ast$ and $V_\ast \times  \left(L_0^2(\Gamma) \times L^2(\Gamma)\right)$, respectively.
It is not clear whether $a(\cdot,\cdot)$ is elliptic on $V_\ast^0$. For well-posedness, however, it is sufficient to have ellipticity of   this bilinear form  on  the kernel of $\tilde b(\cdot,\{\cdot,\cdot\})$:
\[ \mathcal{K}:=\{\, \bu \in V_\ast^0 ~|~ \tilde b(\bu,\{p,\lambda\})=0 \quad \text{for all}~~\{p,\lambda\} \in L_0^2(\Gamma) \times L^2(\Gamma)\,\}.
\]
Note that
\[
  \mathcal{K}\subset \mathcal{K}_{0}:= \{\, \bu \in V_\ast^0 ~|~ \tilde b(\bu,\{0,\lambda\})=0 \quad \text{for all}~~\lambda \in L^2(\Gamma)\,\}= \{\, \bu \in V_\ast^0 ~|~u_N=0\,\}.
\]
Using Lemma~\ref{Kornlemma} it follows that
 \begin{equation} \label{elliptic}
  a(\bu,\bu) = a(\bu_T,\bu_T) \geq  2\mu c_K^2 \|\bu_T\|_1^2 = 2\mu c_K^2 \|\bu\|_{V_\ast}^2  \quad \text{for all}~~\bu \in \mathcal{K}_{0},
 \end{equation}
 and thus we have ellipticity of $a(\cdot,\cdot)$ on the kernel of $\tilde b(\cdot,\{\cdot,\cdot\})$.
 It remains to check the inf-sup condition for $\tilde b(\cdot,\{\cdot,\cdot\})$. Take $\{ p,\lambda \} \in L_0^2(\Gamma) \times L^2(\Gamma)$. Take $\bv_T^0 \in V_T^0$ such that
 \[ b(\bv_T^0,p)=\|p\|_{L^2}^2\quad \tilde{c}\|\bv_T^0\|_1\le \|p\|_{L^2}
 \]
holds, with $\tilde{c}>0$, cf. Lemma~\ref{leminfsup}.
 Take $\bv:=\bv_T^0 + \lambda \bn \in V_\ast^0$, hence
$\|\bv\|_{V_\ast}^2 = \|\bv_T^0\|_1^2 + \|\lambda\|_{L^2(\Gamma)}^2$. We get:
\begin{align*}
 \tilde b(\bv,\{p,\lambda\})&  = b(\bv_T^0,p)+\|\lambda\|_{L^2}^2
 = \|p\|_{L^2}^2+ \|\lambda\|_{L^2}^2
  \\
 & \geq  \min\{1,\tilde{c}\} \big(\|p\|_{L^2}^2
 +\|\lambda\|_{L^2}^2\big)^\frac12 \|\bv\|_{V_\ast}.
\end{align*}
Hence, the required inf-sup property holds, from which the well-posedness result follows.  If in the second equation in \eqref{Stokesweak2} we take $q=0$ and $\nu \in L^2(\Gamma)$ arbitrary, it follows that for the solution $\bu$ we have $u_N=0$, i.e., $\bu \cdot \bn=0$ holds.
\end{proof}
\smallskip

\begin{rem} \label{lambazero} \rm If in the first equation in \eqref{Stokesweak2} we take $v_N=0$, $\bv_T \in E$, it follows from $E_s(\bv_T)=0$, $ \tilde  b(\bv,\{p,\lambda\})=b(\bv_T,p)=0$, $f(\bv)=f(\bv_T)=0$ that the first equation in \eqref{Stokesweak2} is satisfied for all $\bv_T \in E$, hence the test space $V_\ast^0$ can be replaced by $V_\ast$ (which is convenient in a Galerkin method).

For the unique solution $\bu$ we have $u_N=0$,  and taking $v_N=0$, $\nu=0$ it follows that if $f(\bv)=f(\bv_T)$ then $(\bu_T,p)$ coincides with the unique solution of \eqref{Stokesweak1}. In this sense, the problem  \eqref{Stokesweak2} for $(\bu,\{p,\lambda\}) \in V_\ast^0 \times \left(L_0^2(\Gamma) \times L^2(\Gamma)\right)$ is a consistent generalization of   the problem \eqref{Stokesweak1} for $(\bu_T,p) \in V_T^0 \times L_0^2(\Gamma)$. {  Due to the fact that the latter problem is consistent to the original strong formulation, this also holds for the generalized weak formulation \eqref{Stokesweak2}.}
\end{rem}

\section{Well-posed augmented  variational formulations}\label{s_aug}
Another way to relax the tangential constraint in the test and trial spaces is to augment the weak formulation  \eqref{Stokesweak1} with a normal term such that the augmented bilinear form
defines an inner product in $V_\ast$. The augmentation can be done for the bilinear form $a(\cdot,\cdot)$ used in \eqref{Stokesweak1} as well as for the  one
used in \eqref{Stokesweak2}. Given an augmentation parameter $\tau\ge0$, we define
\begin{equation}\label{a_aug}
\begin{split}
% a_\mu(\bu,\bv)&:= \int_\Gamma  \big(E_s( \bu_T): E_s( \bv_T) +\mu (\bu\cdot \bn) (\bv \cdot \bn) \,\big) ds, \\
% \hat{a}_\tau(\bu,\bv)&:= \int_\Gamma\big(  E_s( \bu): E_s( \bv) +\mu (\bu\cdot\bn) (\bv \cdot \bn) \,\big) ds, \\
  a_\tau(\bu,\bv)&:=2\mu \int_\Gamma  E_s( \bu_T): E_s( \bv_T)\, ds  +\tau \int_{\Gamma} u_N v_N \, ds \\ &~ = a( \bu_T,\bv_T)+ \tau (u_N,v_N)_{L^2(\Gamma)},\\
 \hat{a}_\tau(\bu,\bv)&:= 2\mu \int_\Gamma E_s( \bu): E_s( \bv)\, ds+\tau \int_{\Gamma} u_N v_N \, ds \\
 &~ = a( \bu,\bv)+ \tau (u_N,v_N)_{L^2(\Gamma)},
 \end{split}
\end{equation}
for $\bu,\bv \in V_\ast$.
We consider, for $\tau > 0$,  the following two problems: determine $(\bu,p) \in V_\ast^0 \times L_0^2(\Gamma)$ such that
 \begin{equation} \label{Stokesweak1B} (a)~\left\{\begin{split}
           a_\tau(\bu,\bv) +b(\bv_T,p) &=f(\bv_T) , \\
           b(\bu_T,q) & = 0,
                                      \end{split}\right.\quad\text{or}\quad (b)~\left\{\begin{split}
           \hat{a}_\tau(\bu,\bv) +b(\bv_T,p) &=f(\bv_T) , \\
           b(\bu_T,q) & = 0 ,
                                      \end{split}\right.
 \end{equation}
 for all $\bv \in V_\ast$, $q \in L^2(\Gamma)$. Well-posedness of these formulations is given in the following theorem.
 \begin{thm}\label{Th_aug}
 The problem \eqref{Stokesweak1B}{\rm(a)}  is well-posed. The problem \eqref{Stokesweak1B}{\rm(b)} is well-posed for sufficiently large $\tau > 0$. In  \eqref{Stokesweak1B}{\rm(b)} we take $\tau > 0$ sufficiently large such that this problem is well-posed.  The unique solution $\bu$ of \eqref{Stokesweak1B}{\rm(a)}  satisfies $\bu\cdot \bn =0$ and
 $\bu_T$ coincides with the unique solution of \eqref{Stokesweak1}. For the tangential part $\hat\bu_T$ of $\hat\bu$, the unique solution of \eqref{Stokesweak1B}{\rm(b)}, the following estimate holds
 \begin{equation}\label{EstAug}
 \|\hat\bu_T-\bu_T\|_{1}\le C\,\tau^{-\frac12}\|f\|_{V'},
 \end{equation}
 where $C$ depends only on $\Gamma$.
 \end{thm}
 \begin{proof}
  Note that due to Korn's inequality on $V_T^0$ (Lemma~\ref{Kornlemma}) we have
\[
  a_\tau(\bu,\bu) \geq  2\mu c_K^2 \|\bu_T\|_1^2 + \tau \|u_N\|_{L^2}^2 \geq \min \{ 2\mu c_K^2,\tau\} \|\bu\|_{V_\ast}^2.
\]
Hence for  any $\tau > 0$, $a_\tau(\bu,\bv)$ defines a scalar product on $V_\ast^0$.
We already discussed in section~\ref{SecLagrange} that the bilinear form $\hat{a}_\tau(\bu,\bv)$  is well-defined on $V_\ast$  due to the identity \eqref{idfund}.
 If $\tau$ is sufficiently large, for example, ${  \tau> 2\mu}\|\bH\|_{L^\infty(\Gamma)}^2$, then  with the help of triangle and Korn's inequalities
 we get
 \[
 \begin{split}
  \hat{a}_\tau(\bu,\bu)&= 2\mu\|E_s( \bu_T)+u_N\bH\|^2_{L^2}+\tau \|u_N\|^2_{L^2}\ge \mu\|E_s( \bu_T)\|^2_{L^2}-2\mu\|u_N\bH\|^2_{L^2}+\tau \|u_N\|^2_{L^2} \\
  &\ge  {  \mu c_K^2 \|\bu_T\|_1^2 + (\tau-2 \mu\|\bH\|_{L^\infty(\Gamma)}^2) \|u_N\|_{L^2}^2 } \ge c  \|\bu\|_{V_\ast}^2,\quad c>0.
 \end{split}
  \]
 Hence,
 $\hat{a}_\tau(\bu,\bv)$  defines a scalar product on ${ V_\ast^0}$. The inf-sup property for $b(\cdot,\cdot)$ on $V_\ast^0 \times L_0^2(\Gamma)$ immediately follows from the one on $V_T^0 \times  L_0^2(\Gamma)$, i.e., \eqref{infsup}:
  \begin{equation*}
\sup_{\bv\in{V_\ast^0}}\frac{b(\bv_T,p)}{\|\bv\|_{V_\ast}}\underset{ V_T^0\subset V_\ast^0}{\ge}
 \sup_{\bv_T\in{V_T^0}}\frac{b(\bv_T,p)}{\|\bv_T\|_{V_\ast}}
=\sup_{\bv_T\in{V_T^0}}\frac{b(\bv_T,p)}{\|\bv_T\|_1}
  \geq c \|p\|_{L^2},
\end{equation*}
for any $p\in L^2_0(\Gamma)$, with $c>0$ independent of $p$.
The coercivity and continuity of $a$-forms together with continuity and inf-sup property of the $b$-form imply the well-posedness of both problems.  It easy to check that $\bu=\bu_T$, with
 $\bu_T$ the solution of \eqref{Stokesweak1},  solves the augmented problem in \eqref{Stokesweak1B}(a).
 Denote by $\hat\bu,\hat{p}$ the solution of \eqref{Stokesweak1B}(b). By testing the weak formulation with
 $\bv=\hat\bu$, $q=p$, and applying Korn's inequality we obtain the estimate for the normal part of $\hat\bu$,
 \[
 \|\hat{u}_N\|_{L^2(\Gamma)}\le C\tau^{-\frac12}\|f\|_{V'}.
 \]
 For arbitrary $\bv_T\in V_T$ we have thanks to \eqref{idfund}, \eqref{Stokesweak1} and \eqref{Stokesweak1B}(b),
 \begin{multline*}
 a_\tau(\hat\bu_T-\bu_T,\bv_T)=-2\mu\int_\Gamma \hat u_N \bH: E_s( \bv_T)ds+b(\bv_T,p-\hat{p})\\ \le
  C\|\hat{u}_N\|_{L^2(\Gamma)}\|\bv_T\|_{1}+b(\bv_T,p-\hat{p})
  \le C\tau^{-\frac12}\|f\|_{V'}\|\bv_T\|_{1}+b(\bv_T,p-\hat{p}).
 \end{multline*}
 Taking $\bv_T=\hat\bu_T-\bu_T$ the pressure term vanishes and using Korn's inequality for the left-hand side
 leads to \eqref{EstAug}.
 \end{proof}

\smallskip
The well-posedness statements in the theorem above still hold if $f(\bv_T)$ is replaced by $f(\bv)$, with $f \in V_\ast'$.  We close this section with a few remarks.
\smallskip

%\begin{rem} \rm
%The divergence constraint in  \eqref{Stokesweak1B}(a) and  \eqref{Stokesweak1B}(b) can be modified in the sense that $ b(\bu_T,q)$ and $b(\bv_T,p)$
%can be replaced by $ b(\bu,q)$ and $b(\bv,p)$. The latter bilinear forms are defined by (we use same notation $b(\cdot,\cdot)$)
%\begin{equation} \label{balternative}
%b(\bu,p)= - \int_\Gamma p ( \divG \bu_T + \kappa\,u_N) \, ds, \quad\bu \in V_\ast, ~p\in L^2(\Gamma).
%\end{equation}
%If we have sufficient smoothness $\bu \in V$, then $b(\bu,p)= - \int_\Gamma p  \divG \bu \, ds$ holds.
%The modified bilinear form $b(\cdot,\cdot)$ is continuous on $V_\ast \times L^2(\Gamma)$. The inf-sup stability of $b(\cdot,\cdot)$ on $V_\ast^0 \times L_0^2(\Gamma)$ immediately follows from the one on $V_T^0 \times  L_0^2(\Gamma)$, i.e., \eqref{infsup}. Hence, we obtain well-posedness. The property $\bu\cdot \bn=0$ does not necessarily hold.
%\end{rem}
%\smallskip
\begin{rem} \rm \label{rem2a}
We briefly address properties of the different variational formulations  \eqref{Stokesweak1}, \eqref{Stokesweak2} and \eqref{Stokesweak1B} that we consider  relevant for discretization  by Galerkin methods such as fitted or unfitted finite element methods for PDEs posed on surfaces~\cite{DEreview,OlshReusken08}.
 In such finite element methods one ususally approximates a smooth surface $\Gamma$ by a triangulated Lipschitz surface $\Gamma_h$. The normal vector field $\bn_h$ to such a surface is no longer continuous. Enforcing strongly the tangential condition $\bu\cdot\bn_h=0$ for the numerical solution can be inconvenient if standard $H^1(\Gamma)^3$-conforming finite elements are used. Formulations \eqref{Stokesweak2} and \eqref{Stokesweak1B} allow to enforce the tangential condition weakly and occur to us more suitable for  numerical purposes. In \eqref{Stokesweak2} one  needs a suitable finite element space for the Lagrange multiplier $\lambda$. This is avoided in \eqref{Stokesweak1B}, but that formulation  requires a suitable value for the penalty parameter $\tau$. Note that the formulations in   \eqref{Stokesweak2} and \eqref{Stokesweak1B}(a) are consistent with  \eqref{Stokesweak1}, in particular the solution $\bu\in V_\ast^0$ has the property $\bu\cdot \bn=0$. The problem in \eqref{Stokesweak1B}(b) is not
consistent. However, compared to \eqref{Stokesweak1B}(a) the formulation in \eqref{Stokesweak1B}(b) has the attractive property that one has to approximate $\nabla_\Gamma \bu =\bP\nabla \bu\bP$ instead of  $\nabla_\Gamma \bu_T =\bP\nabla \bu_T\bP= \bP \nabla(\bP \bu)\bP$. Hence, in \eqref{Stokesweak1B}(b)  differentiation of $\bP$ is avoided. A finite element discretization for a vector surface Laplace problem (instead of Stokes) based on an  augmented formulation very similar to the one in \eqref{Stokesweak1B}(b) has been studied in the recent paper \cite{hansbo2016analysis}. {  A finite element discretization for a vector surface Laplace problem  based on the saddle point  formulation \eqref{Stokesweak2} has been studied in the recent report \cite{Jankuhn2}}. Finally note that $b(\bu,p)=- \int_{\Gamma} p \divG \bu_T \, ds= -\int_{\Gamma} p\divG (\bP\bu) \, ds$, used  in both \eqref{Stokesweak2} and \eqref{Stokesweak1B}, requires a differentiation of $\bP$. If in the finite element method we
have $p= p_h \in H^1(\Gamma)$ we can use $b(\bu,p)= \int_{\Gamma}\bu_T \nabla_\Gamma p  \, ds$ and thus avoid this differentiation.
\end{rem}

\begin{rem}\label{rem2} \rm  The formal extension of the weak formulations in \eqref{Stokesweak1}, \eqref{Stokesweak2} and \eqref{Stokesweak1B} to the Navier-Stokes equations \eqref{NSstat} on stationary surfaces is  straightforward, but not studied in this paper.
\end{rem}

\section{Conclusions and outlook}\label{sectOutlook} Based on surface mass and momentum conservation laws we derived the surface Navier-Stokes equations \eqref{momentum} and  the corresponding tangential and normal equations.  Similar equations can be found in several other papers in the literature. The equations that we obtain agree with those derived in \cite{Gigaetal} by a completely different approach. All differential operators used are defined in terms of first (partial) derivatives in the outer Euclidean space $\Bbb{R}^3$. Relations to formulations presented in the setting of differential geometry (e.g., Bochner and Hodge-de Rham Laplacians) are briefly addressed. Well-posedness results  of several variational formulations of a Stokes problem on a stationary surface are presented. For this a surface Korn's inequality and  an inf-sup property for the Stokes bilinear form $b(\cdot,\cdot)$ are derived.

{  In the recent  report \cite{Jankuhn2} we present results of numerical experiments for a finite element method applied  to a saddle point formulation of a surface vector-Laplace problem, similar to  \eqref{Stokesweak2}. In forthcoming work, this method will be extended to surface (Navier-)Stokes equations.} Furthermore, we plan to develop error analyses for these finite element discretization methods. Clearly, there are many other related topics that can be addressed in future research. For example, an extension of the well-posedness analysis presented in this paper to the case of a Stokes problem on an evolving surface, the extension from Stokes to an Oseen or  Navier-Stokes equation on a stationary (or even evolving) surface, or an analysis of a coupled surface-bulk flow problem. Related to  the latter we note that first results on well-posedness of such a coupled problem have recently been presented in \cite{lengeler2015stokes}. Furthermore, a further study and validation of such surface
Navier-Stokes equations (coupled with bulk fluids) based on numerical simulations is an open research field.

\subsection*{Acknowledgments}
M.O. was partially supported by NSF through the Division of Mathematical Sciences grant 1522252.

\bibliographystyle{acm}
\bibliography{literatur}{}

\section{Appendix}
We give an elementary proof of the results given in Lemma~\ref{LemGauss}. For this it is very convenient to introduce a tensor notation and the Einstein summation convention for the differential operators $\nabla_i$ (covariant partial derivative) and $\divG$ (surface divergence).
For a scalar function $f$ we have, cf. \eqref{defdeli1}:
\[
  \nabla_i f=\partial_k f P_{ki}=P_{ik}\partial_k f.
\]
(scalar entries of the matrix $\bP$ are denoted $P_{ij}$). For the vector function $\bu: \Rn \to \Rn$ we have, cf. \eqref{defdeli2}:
\[
 \nabla_i u_j:=(\nabla_i \bu)_j=(\nabla_\Gamma \bu)_{ji}=P_{jl}\partial_k u_l P_{ki}=P_{ik}\partial_k u_l P_{lj},
\]
and for matrix valued functions we get, cf. \eqref{defdeli3}:
\[
 \nabla_i A_{sl}:= (\nabla_i \bA)_{sl}= P_{sm}\partial_k A_{mn}P_{nl}P_{ki}= P_{ik}\partial_k A_{mn}P_{ms}P_{nl}.
\]
For the divergence operators we have the representations:
\begin{align*}
 \divG \bu & = (\gradG \bu)_{ii}= P_{ik}\partial_k u_l P_{li}=P_{lk} \partial_k u_l \\
 (\divG \bA)_i & = \divG (e_i^T \bA)= P_{lk}\partial_k A_{il}.
\end{align*}
Below, functions $\bu \in C^2(\Gamma)^n$ are always extended to a neighborhood of $\Gamma$ by taking constant values along the normal $\bn$.
\begin{lem} \label{lemhulp1}
 The following identities hold:
 \begin{align}
  (\bP \divG \gradG^T \bu)_i &= \nabla_k(\gradG \bu)_{ki}=:\nabla_k \nabla_i u_k \label{p1} \\
   \nabla_i (\divG \bu) &= \nabla_i(\gradG \bu)_{kk}=:\nabla_i \nabla_k u_k. \label{p2}
 \end{align}
\end{lem}
\begin{proof}
 We use the representations introduced above and thus get
 \begin{equation} \label{h1}
  (\bP \divG \gradG^T \bu)_i= P_{is} \divG (\gradG^T \bu)_s=P_{is}P_{lk} \partial_k (\gradG \bu)_{ls}.
 \end{equation}
Furthermore,
\[
 \nabla_k(\gradG \bu)_{ki} = P_{kr}\partial_r (\gradG \bu)_{ls}P_{si}P_{lk}= P_{is} P_{lr}\partial_r (\gradG \bu)_{ls},
\]
and comparing this with \eqref{h1} proves the result in \eqref{p1}.
Note that using $P_{lk}n_k=P_{ms}n_m=0$ (where $n_j$ denotes the $j$-th component of the normal vector $\bn$) we get
\[
 (\gradG \bu)_{km}\partial_r P_{mk}= -P_{ms}\partial_s u_l P_{lk}\big((\partial_r n_m) n_k+n_m(\partial_r n_k)\big)=0.
\]
Using this we get
\begin{align}
 \nabla_i(\gradG \bu)_{kk}& = P_{ir}\partial_r (\gradG \bu)_{nm}P_{mk}P_{nk}= P_{ir}\partial_r(\gradG \bu)_{nm} P_{mn}=P_{ir}\partial_r \big((\gradG \bu)_{nm}P_{mn}\big) \nonumber \\
 & = P_{ir}\partial_r (P_{mk} \partial_k u_l P_{ln}P_{mn})= P_{ir}\partial_r (P_{mk}\partial_k u_l P_{lm}) \nonumber \\
 & = P_{ir}\partial_r(P_{lk}\partial_k u_l)= P_{ir}\partial_r (\divG \bu)= \nabla_i(\divG \bu),\label{ppp}
\end{align}
and thus the identity \eqref{p2} holds.
\end{proof}
\ \\[1ex]
We now derive a result for the commutator $\nabla_k \nabla_i u_k-\nabla_i \nabla_k u_k$.
\begin{lem} \label{hulplem2} Let $\bH= \nabla_\Gamma\bn$ be the Weingarten mapping. Then for $\bu \in C^2(\Gamma)^n$ with $\bP\bu=\bu$ the identity
 \[
   \nabla_k \nabla_i u_k-\nabla_i \nabla_k u_k= \big( (\tr(\bH) \bH -\bH^2)\bu\big)_i, \quad i=1,\ldots,n,
 \]
holds.
\end{lem}
\begin{proof}
 By definition we have
 \[
  \nabla_k \nabla_i u_k= P_{kr}\partial_r (\gradG \bu)_{nm}P_{mi}P_{nk}= \partial_r (P_{ms}\partial_s u_l P_{ln}) P_{mi}P_{nr}.
 \]
We use the product rule, $H_{rl}=\partial_r n_l$, $\partial_r P_{ms}=-\partial_r(n_m n_s)=- H_{rm}n_s-H_{rs}n_m$, $P_{mi}n_m=0$,  and thus obtain
\begin{align*}
 \nabla_k \nabla_i u_k &= \big(\partial_r P_{ms} \partial_s u_l P_{ln} +P_{ms} \partial_r \partial_s u_l P_{ln}+ P_{ms}\partial_s u_l \partial_r P_{ln}\big) P_{mi}P_{nr}\\
 & = -H_{rm} n_s \partial_s u_l P_{lr}P_{mi}+P_{is}P_{lr}\partial_s \partial_ru_l -H_{rn} n_l \partial_s u_l P_{is} P_{nr}.
\end{align*}
We also have, cf.~\eqref{ppp},
\begin{align*}
\nabla_i \nabla_k u_k& = P_{ir}\partial_r(P_{lk}\partial_k u_l)= P_{ir}\partial_r P_{lk}\partial_k u_l +P_{ir}P_{lk} \partial_r \partial_k u_l \\
&= -P_{ir}(H_{rl}n_k+H_{rk}n_l)\partial_k u_l+P_{ir}P_{lk} \partial_r \partial_k u_l .
\end{align*}
Hence, for the difference we get
\begin{align*}
&  \nabla_k \nabla_i u_k-\nabla_i \nabla_k u_k \\ & =H_{rl}n_k\partial_k u_lP_{ir} -H_{rm} n_s \partial_s u_lP_{lr}P_{mi}+ H_{rk}n_l\partial_k u_lP_{ir} -H_{rn} n_l \partial_s u_l P_{is} P_{nr}.
\end{align*}
Using $\bP \bu=\bu$ we get
\begin{align*}
 H_{rm} n_s \partial_s u_lP_{lr}P_{mi}& = H_{rm} n_s \partial_s(P_{lr}u_l) P_{mi}- H_{rm} n_su_l \partial_sP_{lr} P_{mi} \\
  &= H_{mr} n_s \partial_s u_r P_{im}- H_{rm} n_s \partial_sP_{lr} P_{mi}u_l.
\end{align*}
Furthermore, using $\bH \bn=0$, we get
\[
 H_{rm} n_s \partial_sP_{lr} P_{mi}u_l= -H_{rm} n_s(H_{sl}n_r +H_{sr}n_l)P_{mi}u_l =0.
\]
Combining these results we get
\[
 \nabla_k \nabla_i u_k-\nabla_i \nabla_k u_k= H_{rk}n_l\partial_k u_lP_{ir} -H_{rn} n_l \partial_s u_l P_{is} P_{nr}.
\]
Using $\bn^T \bu=0$ (in a neighborhood of $\Gamma$) we get $\partial_k(n_lu_l)=0$ and in combination with $\bH\bP=\bP\bH=\bH$ we get
\[ \begin{split}
 H_{rk}n_l\partial_k u_lP_{ir} & = - H_{rk}\partial_k n_lP_{ir} u_l= - H_{rk}H_{kl} P_{ir}u_l \\ & = -H_{ik}H_{kl}u_l= -(H^2)_{il} u_l = -(\bH^2 \bu)_i.
\end{split} \]
Finally note that
\[ \begin{split}
 H_{rn} n_l \partial_s u_l P_{is} P_{nr} & = -H_{rn}\partial_s n_lP_{is} P_{nr} u_l=-H_{rn}H_{sl}P_{is} P_{nr} u_l \\ & =  -H_{rr} H_{il}u_l= - \tr(\bH) (\bH \bu)_i.
\end{split} \]
Combining these results completes the proof.
\end{proof}
\ \\[1ex]
By combining the results of Lemma~\ref{lemhulp1} and Lemma~\ref{hulplem2} we have proved the result \eqref{Res1}. Let $\bA$ be an $n \times n$ matrix with $\bP\bA=\bA\bP=\bA$, hence $\bA \bn= \bA^T \bn=0$. We then have
\begin{align*}
 \bn\cdot \divG \bA&  = n_i (\divG \bA)_i = n_i P_{lk}\partial_k A_{il} = P_{lk}\partial_k (n_iA_{il})-P_{lk} \partial_k n_i A_{il} \\
 &=- P_{lk} H_{ki}A_{il} = - H_{li} A_{il}= -(\bH \bA)_{ll}= -\tr(\bH\bA),
\end{align*}
and combining this with $\tr(\bH \bA)=\tr(\bA^T \bH)= \tr(\bH \bA^T)$ one obtains the result in \eqref{Res1a}. The result in \eqref{newres} follows from (we use $\bH \bn =0$, $\bn^T \bH=0$):
\begin{align*}
 (\bP\divG (\bH))_i & = P_{ij}P_{lk}\partial_k H_{jl} = P_{ij} P_{lk} \partial_k \partial_j n_l = P_{ij} P_{lk} \partial_j \partial_k n_l =P_{ij}P_{lk}\partial_j H_{kl} \\
& = P_{ij} \partial_j (P_{lk}H_{kl}) - P_{ij} (\partial_j P_{lk})H_{kl}  \\ & =P_{ij}  \partial_j H_{ll} + P_{ij}\big( (\partial_j n_l)n_k + (\partial_j n_k)n_l\big) H_{kl}
 =  P_{ij} \partial_j \kappa = (\nabla_\Gamma \kappa)_i.
\end{align*}

\begin{lem} \label{LLem} For $n=3$ the identity
 \[ \tr(\bH) \bH -\bH^2= K\bP,
 \]
with $K$ the Gauss curvature, holds.
\end{lem}
\begin{proof}
We apply the Cayley-Hamilton theorem to the linear mapping $\bP \bH \bP=\bH:\, {\rm range}(\bP) \to {\rm range}(\bP)$. Note that $\dim( {\rm range}(\bP))=2$.  This yields
\begin{equation} \label{eq67}
  \bH^2 -\tr(\bH)\bH +\det(\bH)\bP=0,
\end{equation}
and using $\det(\bH)=K$ we obtain
%Let $\kappa_1,\kappa_2$ be the principal curvatures, with orthonormal eigenvectors $\bH \bv_i= \kappa_i\bv_i$, $i=1,2$, and $K=\kappa_1 \kappa_2$. Define $\bV=(\bv_1~\bv_2~\bn)$, $\bD={\rm diag}(\kappa_1,\kappa_2,0)$. Now note
%\begin{align*}
% \tr(\bH) \bH -\bH^2& = \bV ((\kappa_1+\kappa_2)\bD - \bD^2)\bV^{-1} = \bV {\rm diag}(\kappa_1\kappa_2,\kappa_1\kappa_2,0)\bV^{-1} \\
%  & = K (\bv_1\bv_1^T +\bv_2\bv_2^T)= K \bP,
%\end{align*}
the desired result.
\end{proof}
\ \\[1ex]
The result in \eqref{Res2} follows from \eqref{Res1} and Lemma~\ref{LLem}. As a corollary of \eqref{eq67} we obtain for $n=3$ the identity
\begin{equation} \label{iddd}
 \kappa \bP-\bH =K \bH^\dagger,
\end{equation}
where $\bH^\dagger$ is the generalized inverse of  $\bH$. Note that $K \bH^\dagger$ has the same eigenvalues and eigenvectors as $\bH$, but the eigenpairs are not the same.

\end{document}